\documentclass[12pt, journal,onecolumn]{IEEEtran}
\usepackage{enumerate}   
\usepackage[dvips]{graphicx}
\usepackage{epsfig}
\usepackage[cmex10]{amsmath}
\usepackage{amssymb}
\usepackage{amsthm}
\usepackage{amsfonts}
\usepackage{dsfont}
\usepackage{bm}
\usepackage{cite}
\usepackage[svgnames]{xcolor} 
\usepackage{pstricks,pst-node,pst-plot,pstricks-add}
\usepackage[tight,footnotesize]{subfigure}
\usepackage{algpseudocode}
\usepackage{algorithm}
\usepackage{nccmath}
\usepackage[OT1]{fontenc} 
\usepackage{xcolor}
\newenvironment{mpmatrix}{\begin{medsize}\begin{bmatrix}}%
		{\end{bmatrix}\end{medsize}}%
\graphicspath{{figs/}}
\input{niesen.def}

\begin{document}
\title{Secure Coded Multi-Party Computation for Massive Matrix Operations
\footnote{This is an extended version of the paper, partially presented in IEEE Communication Theory Workshop (CTW), May 2018, and IEEE International Symposium on Information Theory (ISIT), June 2018~\cite{Multi-Party}.
}}

\author{Hanzaleh~Akbari~Nodehi$^{*}$, Mohammad~Ali~Maddah-Ali$^{\dagger}$%
	\\{$^{*}$Department of Electrical Engineering, Sharif University of Technology}%
	\\{$^{\dagger}$Nokia Bell Labs}%
}

\maketitle

\begin{abstract} 
In this paper, we consider a secure multi-party computation problem (MPC), where the goal is to offload the computation of an arbitrary polynomial function of some massive private matrices (inputs) to a cluster of workers. The workers are not reliable. Some of them may collude to gain information about the input data (semi-honest workers). 
The system is initialized by sharing a (randomized) function of each input matrix to each server.  Since the input matrices are massive, each share's size is assumed to be at most $1/k$ fraction of the input matrix, for some $k \in \mathbb{N}$.   The objective is to minimize the number of workers needed to perform the computation task correctly, such that even if an arbitrary subset of $t-1$ workers, for some $t \in \mathbb{N}$, collude, they cannot gain any information about the input matrices. We propose a sharing scheme, called \emph{polynomial sharing}, and show that it admits basic operations such as adding and multiplication of matrices and transposing a matrix. By concatenating the procedures for basic operations, we show that any polynomial function of the input matrices can be calculated,  subject to the problem constraints. We show that the proposed scheme can offer order-wise gain in terms of the number of workers needed, compared to the approaches formed by the concatenation of job splitting and conventional MPC approaches.
\end{abstract}

\section{Introduction}
\label{sec:introduction}

With the growing size of datasets in use cases such as machine learning and data science, it is inevitable to distribute the computation tasks to some external entities, which are not necessarily trusted. In this set-up,  some of the major challenges are protecting the privacy of the data, guaranteeing the correctness of the result, and ensuring the efficiency of the computation. 

The problem of processing private information on some external parties has been studied in the context of \emph{secure multi-party computation (MPC)}. Informally, in an MPC problem, some private data inputs are available in some source nodes, and the goal is to offload the computation of a specific function of those inputs to some parties (workers). These parties are not reliable. Some of them may collude to gain information about the private inputs (semi-honest workers) or even behave adversarially to make the result incorrect. Thus, the objective is to design a scheme, probably based on coding and randomization techniques, to ensure data privacy and correctness of the result.  To ensure privacy, some MPC solutions, such as ~\cite{BMR90, GRR98},  rely on cryptographic hardness assumptions, while others, such as~\cite{ben1988completeness, CCD88,beaver1991efficient},  are protected based on information-theoretic measures  (See~\cite{MPC_BOOK} for a survey on different approaches of MPC). In particular, the BGW scheme, named after its inventors,  Ben-Or, Goldwasser, and Wigderson in \cite{ben1988completeness},   relies on Shamir secret sharing~\cite{shamir1979share1} to develop an information-theoretically private MPC scheme to calculate any polynomial of private inputs.  Shamir secret sharing is an approach that allows us to share a secret, i.e., private input,  among some parties, such that if the number of colluding nodes is less than a threshold, they cannot gain any information about the data. The BGW scheme exploits the fact that Shamir secret sharing admits basic operations such as addition and multiplication (at the cost of some communication among nodes).

There have been some efforts to improve MPC algorithms' efficiency, but mainly focusing on the communication loads (see~\cite{scalablemulti-party}). However, less effort has been dedicated to the cases where the input data is massive. One approach would be to split the computation to some smaller subtasks and dedicate a group of workers to execute each subtask, using conventional MPC approaches. In this paper, we argue that the idea of the concatenation of job splitting and multi-party computation could be significantly sub-optimal in terms of the number of workers needed. 

In a seemingly irrelevant area, extensive efforts have been dedicated to using coding theory to improve the efficiency of distributed computing,  mainly to cope with the stragglers~\cite{dean2013tail, speed, high,yu2017polynomial, entangle, opt-recovery, short, sketch, sparse, lagrange,  codedsketch}. The core idea is based on partitioning each input data into some smaller inputs and then encoding the smaller inputs. The workers then process those encoded inputs. The ultimate goal is to design the code such that the computation per worker node is limited to what it can handle, and also the final result can be derived from the outcomes of a subset of worker nodes. This is done for matrix to vector multiplication in~\cite{speed}, and matrix to matrix multiplication in~\cite{ high,yu2017polynomial,entangle, opt-recovery, short}. In particular, in~\cite{yu2017polynomial}, the code is designed such that the results of different workers form a maximum distance separable (MDS), meaning that the final result can be recovered from any subset of servers with the minimum size. That approach has been extended to general matrix partitioning  in~\cite{entangle, opt-recovery}, and also for the cases where only an approximate result of the matrix multiplication is needed in~\cite{sketch,codedsketch}.
This paper aims to design efficient MPC schemes for massive inputs, exploiting the ideas developed in the context of coding for computing.  Let us first review a motivating example that justifies our objective. 

 \begin{example}[Maximum likelihood regression on private data:]
		\label{example:reg}
		Assume that a research center wants to run a machine learning algorithm, say maximum likelihood regression, on a collection of private data sets, where each data set belongs to an organization.  For example, each data set can be the salary information of a company (see~\cite{WageGap}) or the medical records of a hospital. Those organizations want to help the research center, but do not want to reveal any information about their data sets beyond the final result. In particular, assume that there are $\Gamma \in \mathbb{N}$ organizations, where organization $\gamma \in \{ 1, 2, \ldots, \Gamma\}$,  has  a matrix $\mathbf{X}^{[\gamma]}$ and a  corresponding target vector $\mathbf{b}^{[\gamma]}$. The target value of each row of  $\mathbf{X}^{[\gamma]}$ is the corresponding entry of  vector $\mathbf{b}^{[\gamma]}$. The research center aims to find a regression model  $b= \mathbf{w}^T \mathbf{x}$, representing the relationship between each row of $\mathbf{X}$ and the corresponding entry of $\mathbf{b}$, where
		
		\begin{align}
		\mathbf{X}=\left[
		\begin{matrix}
		\mathbf{X}^{[1]}\\
		\vdots\\
		\mathbf{X}^{[\Gamma]}
		\end{matrix}
		\right], \ \mathbf{b}=\left[
		\begin{matrix}
		\mathbf{b}^{[1]}\\
		\vdots\\
		\mathbf{b}^{[\Gamma]}
		\end{matrix}
		\right]. 
		\end{align}
		Following the maximum likelihood solution for regression (see~\cite{bishopbook}[Chapter 3.1.1],  the center needs to calculate 
		\begin{align}
		\mathbf{w}= (\mathbf{X}^T\mathbf{X})^{-1} \mathbf{X}^T  \mathbf{y},
		\end{align} 
		where  $(\mathbf{X}^T\mathbf{X})^{-1} \approx \mathbf{I}+ \mathbf{A}+\ldots+\mathbf{A}^k$, 
		for $\mathbf{A}= \mathbf{I}-\mathbf{X}^T\mathbf{X}$ and some $k \in \mathbb{N}$ (assuming that the datasets are normalized such that the largest eigenvalue of $\mathbf{A}$ is less than one). For processing, assume that there are some servers available, where none of them has enough computing and storage resources to execute this computation alone. On top of that, at most $t-1$, for some integer $t$,  of them may collude to gain information about the private data.  Conventional MPC schemes, such as the BGW scheme~\cite{ben1988completeness}, can handle computing polynomial functions of some private inputs (See~\cite{MPC_BOOK}), but are not designed to handle scenarios where the inputs are massive matrices. In this paper, we aim to address such scenarios. 
	\end{example}

In this paper we consider a system, including $\Gamma$ sources, $N$ workers, and one master. There is a link between each source and each worker. All of the workers are connected to each other and also are connected to the master.  Each source sends a function of its data (so-called a share) to each worker. 
We assume that the workers have limited computation resources. As a proxy to that limit, we assume that each share's size can be up to a certain fraction of the corresponding input size. The workers process their inputs, and in between, they may communicate with each other.  After that, each worker sends a message to the master, such that the master can recover the required function of the inputs. The sharing and the computation procedures must be designed such that if any subset of $t-1$ workers collude, for some $t \in \mathbb{N}$,  they can not gain any information about the inputs. Also, the master must not gain any additional information, beyond the result,  about the inputs. Motivated by recent results in coding for matrix multiplication, as an extension to Shamir secret sharing,  in this paper, we propose a new sharing approach called \emph{polynomial sharing}. We show that the proposed sharing approach admits basic operations such as addition, multiplication, and transposing, by developing a procedure for each of them.  Finally, we prove that we can compute any polynomial function using these procedures while preserving the privacy subject to the storage limit of each worker node. We show that the number of servers needed to compute a function using this approach is order-wise less than what we need in approaches based on job splitting and the conventional BGW scheme. 

The rest of the paper is organized as follows.  In Section \ref{sec:Problem Setting}, we formally state the problem setting. In Section \ref{sec:preliminaries}, we review some preliminaries and conventional approaches for MPC. In section \ref{main result}, we state the main result. In Section \ref{Motivating Examples for proposed scheme}, we review some motivating examples. In Section \ref{Polynomial Sharing}, we present the \emph{polynomial sharing} scheme.  In Section \ref{Procedures}, we show several procedures to perform basic operations, such as addition, multiplication, and transposing, using the proposed sharing scheme. In Section \ref{LSMPC Algorithm}, we present the algorithm to calculate general polynomials. Finally in Section \ref{Extension}, we present some extensions.

\emph{Notation:} In this paper matrices and vectors are denoted by upper boldface letters and lower boldface letters respectively. For $n_1, n_2\in\mathbb{Z}$ the notation $[n_1,n_2]$ represents the set $\{n_1, n_1 + 1, \dots n_2\}$. Also, $[n]$ denotes the set $\{ 1, \dots,n\}$ for $n\in\mathbb{N}$. Furthermore, the cardinality of a set $\mathcal{S}$ is denoted by $|\mathcal{S}|$. For the arbitrary field $\mathbb{F}$, the notation $\mathbb{F}^*$ means any matrix with any possible size, with entries from $\mathbb{F}$. For a matrix $\mathbf{A}$, $\mathbf{A}(a:b,:)$ denotes a matrix including rows $a$ to $b$ of matrix $\mathbf{A}$.

\subsection{Concurrent and Follow-up Results}
The early version of this paper was submitted to IEEE ISIT 2018 in Jan. 2018, which has been appeared in June 2018~\cite{Multi-Party}. In addition, it was presented in CTW 2018 in May 2018. We have also presented a generalized version of~\cite{Multi-Party} in~\cite{Multi-Party2}. 
	
	In parallel in~\cite{lagrange} and \cite{yu2018coding}, the authors introduce Lagrange coded computing, targeting the case where the computation of interest can be decomposed into computing one polynomial function for several  (private) inputs.  In other words, the computing task can be stated as computing   $\tilde{\mathbf{G}}(\mathbf{X}^{[1]}), \tilde{\mathbf{G}}(\mathbf{X}^{[2]}),  \dots, \tilde{\mathbf{G}}(\mathbf{X}^{[\Gamma]})$ for an arbitrary polynomial function $\tilde{\mathbf{G}}(.)$, and   inputs  $\mathbf{X}^{[1]}, \mathbf{X}^{[2]} \dots, \mathbf{X}^{[\Gamma]}$. The decomposition must be such that one worker can handle one computing of $\tilde{\mathbf{G}}(.)$ for an input. The idea is then to use coding over those computations to form coded redundancy to deal with stragglers and/or guarantee privacy. 
	
	
	The major difference between \cite{lagrange}, \cite{yu2018coding}, and what we do in this paper is that  in \cite{lagrange} and \cite{yu2018coding}, one server \emph{can} handle computing   $\tilde{\mathbf{G}}(\mathbf{X}^{[\gamma]})$, for the input $\mathbf{X}^{[\gamma]}$,   $\gamma \in [\Gamma]$. However,  in this work, the objective is to compute $\mathbf{G}(\mathbf{X}^{[1]}, \mathbf{X}^{[2]} \dots, \mathbf{X}^{[\Gamma]})$, for an arbitrary polynomial $\mathbf{G}$, where one server cannot handle computing $\mathbf{G}(\mathbf{X}^{[1]}, \mathbf{X}^{[2]} \dots, \mathbf{X}^{[\Gamma]})$ alone. It is clear that for an arbitrary polynomial function $\mathbf{G}(\mathbf{X}^{[1]}, \mathbf{X}^{[2]} \dots, \mathbf{X}^{[\Gamma]})$, we cannot always transform it into computing some independent calculations $\tilde{\mathbf{G}}(\mathbf{X}^{[1]}), \tilde{\mathbf{G}}(\mathbf{X}^{[2]}) \dots, \tilde{\mathbf{G}}(\mathbf{X}^{[\Gamma]})$ (see Example~\ref{example:reg}). 
	
	Another difference is that  in \cite{lagrange} and \cite{yu2018coding}, there is no communication among  the workers and the number of workers needed for coding to be effective, is at least  $(\Gamma-1) \deg(\tilde{\mathbf{G}})+1$. However, in the scheme that we propose in this paper, there is communication among the workers and the number of servers needed does not grow with $\Gamma$ or the degree of $\mathbf{G}$. 
	
	The idea of \cite{lagrange} and \cite{yu2018coding} has been used to train some machine learning algorithms~\cite{so2019codedprivateml}, in which the computation can be decomposed into calculating one specific function for several inputs.  To avoid requiring too many workers,   in ~\cite{so2019codedprivateml}, it is suggested to approximate the specific function to a low degree polynomial function.
	
	Another related line of research, started by~\cite{chang2018capacity}, published in June 2018 on arXiv, is known as secure matrix multiplication. In that set-up, the objective is to calculate the multiplication of two matrices  without leaking information to the workers. The scheme of~\cite{chang2018capacity} was later improved by the follow-up results~\cite{KES19, d2020gasp}.  The difference between what we do in this paper and secure matrix multiplication~\cite{chang2018capacity, KES19, d2020gasp} is that we are interested in the result of a general polynomial function of multiple private inputs, while in secure matrix multiplication, the master is interested in the result of multiplication of two matrices. Also, there is no privacy constraint for the master at the secure matrix multiplication, and there is no communication between the workers. Still, we can consider the problem of secure matrix multiplication as a special case of the proposed scenario in this paper, ignoring the extra privacy constraint that we have at the master. Indeed, the scheme that we had already proposed in~\cite{Multi-Party} outperforms the scheme of the concurrent work~\cite{chang2018capacity} and the follow-up work~\cite{KES19}. The other follow-up paper \cite{d2020gasp} reports two schemes, named as  GAPS-Big and GAPS-Small. Indeed, GAPS-Big is the same as what we had already reported in~\cite{Multi-Party}.  The scheme GAPS-Small performs better than what we report in this paper for the case of $3 \leq t < k$. 
	
	Some recent result~\cite{jia2019capacity} on secure matrix multiplications focuses on calculating pairwise multiplications of \emph{several pairs} of matrices, rather than just one pair. This will reduce the cost of randomization per pair of multiplication (using ideas such as cross-subspace alignment) and improve the efficiency~\cite{jia2019capacity}. \cite{kim2019private} considers private matrix multiplications where workers do not collude, but still the master wants to keep the input matrices private from each worker. \cite{kim2019private} also considers the case where one of the matrices is selected from a finite and known set, and the objective is to keep the index of that matrix private.  In \cite{aliasgari2020private}, the authors propose a code for private matrix multiplication that is flexible to achieve a  trade-off between the number of workers needed and the communication load. They also consider the case where one of the matrices is selected from a finite and publicity known set of matrices.  A different approach to privacy in computing is introduced in~\cite{Tahmasebi_Seq_ISIT2019, Tahmasebi_Seq_ISIT2020}, where the function of interest is formed by a specific concatenation and combination of several known linear functions, represented by matrix operations,  and the objective is to keep the order of concatenation/combination private.

\section{Problem Setting}
\label{sec:Problem Setting}
Consider an  MPC system including $\Gamma$ source nodes,  $N$ worker nodes, and one master node, for some $\Gamma, N \in \mathbb{N}$  (see Fig. \ref{shapeofSystem}). 

Each source is connected to every single worker. In addition, every pair of workers are connected to each other.  However, there is no communication link between the sources. In addition, all of the workers are connected to the master. All of these links are orthogonal\footnote{By orthogonal link, we mean that they do not interfere with each other. For example, they can be wired links, or if they are wireless, they communicate at different time/frequency slots.}, secure,  and error free.

Each source $\gamma \in [\Gamma]$ has access to a matrix $\mathbf{X}^{[\gamma]}$, chosen  from an arbitrary distribution over $\mathbb{F}^{m \times m}$ for some finite field $\mathbb{F}$ and $m \in \mathbb{N}$. The master aims to know the result of a function  $\mathbf{Y} = \mathbf{G}(\mathbf{X}^{[1]}, \mathbf{X}^{[2]} \dots, \mathbf{X}^{[\Gamma]})$, where $\mathbf{G} \from (\mathbb{F}^{m \times m})^\Gamma \to \mathbb{F}^{m \times m}$ is an arbitrary polynomial function. The system operates in three phases: (i) sharing, (ii) computation and communication, and (iii) reconstruction. The detailed description of these phases is as follows.
\begin{enumerate}
	\item \emph{Sharing}: In this phase, the source $\gamma$ sends $\mathbf{\tilde{X}}_{\gamma n} = \mathbf{F}_{\gamma \to n}(\mathbf{X}_{\gamma})$ to worker $n$  where $	\mathbf{F}_{\gamma \to n} \from \mathbb{F}^{m \times m}\to \mathbb{F}^{m \times \frac{m}{k}}$, for some $k \in \mathbb{N}$, $k|m$, denotes the sharing function,  used at source  $\gamma \in [\Gamma]$, to share data with worker $n \in [N]$.  The number $k$ represents the limit on the storage size at each worker.

	\item \emph{Computation and Communication}: In this phase, the workers process what they received, and in between, they may send some messages to other workers and continue processing.  We define the set $\mc{M}_{n \to n^{\prime}} \in  \mathbb{F}^* $  as the set of all  messages that worker $n$ sends to the worker ${n^{\prime}}$ in this phase, for $n, n^{\prime} \in [N] $.
	
	\item \emph{Reconstruction}: In this phase, every worker sends a message to the master. More precisely, worker $n$ sends the message $\mathbf{O}_n \in  \mathbb{F}^*$ to the master.
\end{enumerate}

This scheme must satisfy three constraints.

\begin{enumerate}
	\item \emph{Correctness}: The master must be able to recover $\mathbf{Y}$ from $\mathbf{O}_1, \mathbf{O}_2,\dots, \mathbf{O}_N$. More precisely
	\begin{align}\label{correctness}
	H(\mathbf{Y}|\mathbf{O}_1, \mathbf{O}_2,\dots, \mathbf{O}_N) = 0,
	\end{align}
	where $H$ denotes the Shannon entropy.
	\item \emph{Privacy at the workers}: Let $t \in [N]$. Any arbitrary subset of workers, including $t-1$ workers, must not gain any information about the inputs. In particular, for any $\mc{S} \subset [N]$, $|\mc{S}| \leq t-1$
	\begin{align}\label{privacy for workers}
	H(\mathbf{X}^{[j]}, j\in [\Gamma]| \bigcup_{n \in \mc{S}} \{\mc{M}_{n^{\prime} \to n}, n^{\prime} \in [N]\}, \mathbf{\tilde{X}}_{\gamma n}, \gamma \in [\Gamma], n \in \mc{S})  = H(\mathbf{X}^{[j]}, j\in [\Gamma]). 
	\end{align}
	$t$ is called \emph{the security threshold} of the system.  In other words, we assume that there are $t-1$ semi-honest worker nodes among the workers. It means that even-though they follow the protocol and report any calculations correctly, they are curious about the input data and may collude to gain information about it. Condition \eqref{privacy for workers} guarantees that those colluding servers gain no information about the private inputs.
	
	\item \emph{Privacy at the master}: The master must not gain any additional information about the inputs, beyond the result of the function. In other words, $\mathbf{Y}$ is the only new information that is revealed to the master. More precisely, 
	\begin{align}\label{privacy for master}
	H(\mathbf{X}^{[1]}, \mathbf{X}^{[2]} \dots, \mathbf{X}^{[\Gamma]}   |\mathbf{Y},\mathbf{O}_1, \mathbf{O}_2,\dots, \mathbf{O}_N)     =
	H(\mathbf{X}^{[1]}, \mathbf{X}^{[2]} \dots, \mathbf{X}^{[\Gamma]}|\mathbf{Y}). 
	\end{align} 
\end{enumerate}

\begin{definition}
	For some $k,t \in \mathbb{N}$ and polynomial function $\mathbf{G}$, we define $N_{\mathbf{G}}^\star(t,k)$ as the minimum number of workers needed to calculate $\mathbf{G}$, while the correctness,  privacy at the workers, and privacy at the master are satisfied.
\end{definition}
The objective of this paper is to find an upper bound on $N_{\mathbf{G}}^\star(t,k)$.

\begin{figure}[!htbp]
	\centering
\includegraphics[width=8cm]{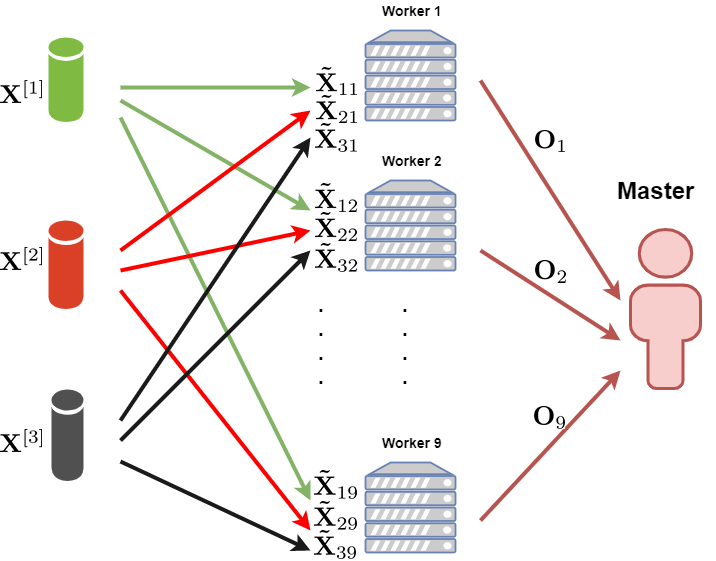}
	\caption{An MPC system including  $\Gamma = 3$ private inputs $\mathbf{X}^{[1]},\mathbf{X}^{[2]}$, and $\mathbf{X}^{[3]} \in \mathbb{F}^{m \times m}$, for $m \in \mathbb{N}$, $N = 9$ workers, and a master. All communication links are secure and error free. The size of the shares given to each worker	is a fraction of the size of the inputs. Input node $\gamma$ sends $\mathbf{\tilde{X}}_{\gamma n} \in \mathbb{F}^{m \times \frac{m}{k}}$ to worker $n$, for some $k \in \mathbb{N}, k|m$.  Workers process their inputs while interacting with each other. Finally worker $n$ sends $\mathbf{O}_n$ to the master. The master aims to know the result of a function e.g., $\mathbf{G}(\mathbf{X}^{[1]}, \mathbf{X}^{[2]}, \mathbf{X}^{[3]}) = (\mathbf{X}^{[1]})^T\mathbf{X}^{[2]}+\mathbf{X}^{[3]}$, subject to the correctness condition \eqref{correctness} and privacy conditions \eqref{privacy for workers}, \eqref{privacy for master}.}
	\label{shapeofSystem}
\end{figure}

\section{Preliminaries: $k=1$}
\label{sec:preliminaries}
As described in Section \ref{sec:Problem Setting}, we have four constraints: limited storage size at each worker represented by $k$, the correctness of the result, and the privacy at the workers and master. If there is no storage limit at the workers, i.e., $k = 1$, the problem is reduced to a version of secure multi-party computation, which has been extensively studied in the literature. In Particular, in \cite{ben1988completeness}, Ben-Or, Goldwasser, and Wigderson propose a scheme referred as the BGW scheme. To be self contained, here we briefly describe the BGW scheme in several examples. As we will see having the limit of $k > 1$ will drastically change the problem.

\begin{example}[\emph{The BGW scheme for addition}]
	\label{avalin mesal}
	Assume that we have two sources $1$ and $2$ with private inputs $\mathbf{X}^{[1]}= \mathbf{A}\in \mathbb{F}^ {m \times m}$ and  $\mathbf{X}^{[2]}= \mathbf{B} \in \mathbb{F}^ {m \times m}$, respectively. These sources share their inputs with the workers. There are at most $t-1$ semi-honest adversaries among the workers, where $t \in [N]$,  meaning that they follow the protocol but may collude to gain information about the input data.

The BGW protocol for this problem is as follows:

\begin{enumerate}
	\item \emph{Phase 1 - Sharing:}
	
	First phase of the BGW scheme is based on Shamir secret sharing \cite{shamir1979share1}. Source $1$ forms the polynomial
	
	\begin{align}
	\label{FA}
		\mathbf{F}_\Ab(x) = \mathbf{A} + \mathbf{\bar{A}}_1x + \mathbf{\bar{A}}_2x^2 + \dots+ \mathbf{\bar{A}}_{t-1}x^{t-1} ,
	\end{align}
	where $\mathbf{F}_\Ab(0) = \mathbf{A}$ is the private input at source 1 and the other coefficients, $\mathbf{\bar{A}}_i$, $i \in \{ 1, 2, \dots t-1\}$,  are chosen from $\mathbb{F}^{m \times m}$ independently and uniformly at random. Source 1 uses  $\mathbf{F}_\Ab(x)$ to share $\mathbf{A}$ with the workers. This means that it sends $\mathbf{F}_\Ab(\alpha_n)$ to worker $n$, for some distinct $\alpha_1, \alpha_2, \dots, \alpha_N \in \mathbb{F}$. 

We note that from Lagrange interpolation rule \cite{bakhvalov1977numerical}, if we have $t$  points from $\mathbf{F}_\Ab(x)$ we can uniquely determine this polynomial~\cite{bakhvalov1977numerical}. Therefore, any subset of including $t$ workers collaboratively can reconstruct $\mathbf{F}_\Ab(0) = \mathbf{A}$. 
Intuitively, the reason is that there are $t$ unknown coefficients in $\mathbf{F}_\Ab(x)$ and thus, we need at least $t$ equations to solve for those coefficients. Also, since coefficients, i.e., $\mathbf{\bar{A}}_i, i \in \{ 1, 2, \dots t-1\}$, are chosen in $\mathbb{F}^{m \times m}$ independently and uniformly at random, then any subset of including less than $t$ workers cannot reconstruct $\mathbf{A}$ and indeed gain no information about it. 	 
 For formal proof, see~\cite{shamir1979share1}.

This approach is called $(N,t)$ sharing of $\mathbf{A}$ (or interchangeably $(N,t)$ Shamir secret sharing of $\mathbf{A}$) and $\mathbf{F}_\Ab(\alpha_n)$ is called the share of $\mathbf{A}$ at worker $n$.
	
	Similarly, source 2 forms $\mathbf{F}_\Bb(x)$ according to the following equation and shares its secrets among the workers,
	\begin{align}
	\label{FB}
		\mathbf{F}_\Bb(x) = \mathbf{B} + \mathbf{\bar{B}}_1x + \mathbf{\bar{B}}_2x^2+ \dots+ \mathbf{\bar{B}}_{t-1}x^{t-1}, 
	\end{align}
	where $\mathbf{F}_\Bb(0) = \mathbf{B}$ is the private input at source 2 and $\mathbf{\bar{B}}_i$, $i \in \{ 1, 2, \dots t-1\}$, are chosen in $\mathbb{F}^{m \times m}$ independently and uniformly at random. Source 2 sends $\mathbf{F}_\Bb(\alpha_n)$ to worker $n$.
	
		\item \emph{Phase 2 - Computation and Communication:}	
		
	Shamir secret sharing scheme has the linearity property, which is very important. It means that if we have the shared secrets $\mathbf{A}$ and $\mathbf{B}$ using $\mathbf{F}_\Ab(x)$ and $\mathbf{F}_\Bb(x)$, respectively among $N$ workers, in order to share the secret $p\mathbf{A}+q\mathbf{B}$ among the workers, for some constants  $p,q \in \mathbb{F}$, we just need that worker $n$ locally calculates $p\mathbf{F}_\Ab(\alpha_n) + q\mathbf{F}_\Bb(\alpha_n)$. This implies that the share of $p\mathbf{A}+q\mathbf{B}$ are available at the workers using sharing polynomial $p\mathbf{F}_\Ab(x)+q\mathbf{F}_\Bb(x)$. For addition, it is sufficient to choose $p = q = 1$.
	Note that in this particular example, no communication among the nodes is needed. 
	
		\item \emph{Phase 3 - Reconstruction:}

	 In this phase  worker $n$ sends $\mathbf{O}_n = \mathbf{F}_\Ab(\alpha_n) + \mathbf{F}_\Bb(\alpha_n)$, calculated in phase two, to the master.  If the master has  $\mathbf{F}_\Ab(\alpha_n) + \mathbf{F}_\Bb(\alpha_n)$ for $t$ or more distinct $\alpha_n$'s, it can recover all the coefficients of degree $t-1$ polynomial  $\mathbf{F}_\Ab(x) + \mathbf{F}_\Bb(x)$. In particular, it recovers  $\mathbf{Y} = \mathbf{F}_\Ab(0) + \mathbf{F}_\Bb(0)=\mathbf{A}+\mathbf{B}$. One can verify that both privacy constraints \eqref{privacy for workers} and \eqref{privacy for master} are satisfied.
	\end{enumerate}

\end{example}

\begin{example}[\emph{The BGW scheme for multiplication}]
	\label{mesale2}
	Assume that we have two sources $1$ and $2$ with private inputs $\mathbf{X}^{[1]}= \mathbf{A}\in \mathbb{F}^ {m \times m}$ and  $\mathbf{X}^{[2]} = \mathbf{B} \in \mathbb{F}^ {m \times m}$, respectively. In addition, the master aims to know $\mathbf{Y}=\mathbf{A}^T\mathbf{B}$.  There are $t-1$ semi-honest adversaries among the $N$ workers. 
	
The BGW protocol for this problem operates as follows:
	
	\begin{enumerate}
		\item \emph{Phase 1 - Sharing:}
		
		This phase is exactly the same as the first phase of Example~\ref{avalin mesal}. Therefore, at the end of this phase worker $n$ has $\mathbf{F}_\Ab(\alpha_n)$ and $\mathbf{F}_\Bb(\alpha_n)$ for some distinct $\alpha_1, \alpha_2, \dots, \alpha_N \in \mathbb{F}$, where $\mathbf{F}_\Ab(x)$ and $\mathbf{F}_\Bb(x)$ are defined in~\eqref{FA} and \eqref{FB}, respectively. 
	
		\item \emph{Phase 2 - Computation and Communication}:

	In this phase, worker $n$ calculates $\mathbf{F}^T_\Ab(\alpha_n)\mathbf{F}_\Bb(\alpha_n)$, simply by multiplying its shares of $\mathbf{A}$ and $\mathbf{B}$. Let us define the polynomial $\mathbf{H}(x) \triangleq \mathbf{F}_{\Ab}^T(x)\mathbf{F}_\Bb(x)$. We note that  $\mathbf{H}(0) = \mathbf{A}^T\mathbf{B}$ and $\deg(\mathbf{H}(x)) = 2t-2$. Therefore,  by having at least $2t-1$ samples of this polynomial,   $\mathbf{A}^T\mathbf{B}$ can be solved for. Therefore, if $N \geq 2t-1$, and these workers send their result to the master, it can calculate $\mathbf{A}^T\mathbf{B}$. However, we do not do that. For some important reasons that we explain later, in Remark \eqref{remarkfor Shamir}, we prefer to first have the Shamir sharing of $\mathbf{A}^T\mathbf{B}$ at each node. As explained before, different workers have samples of $\mathbf{F}_{\Ab}^T(x)\mathbf{F}_\Bb(x)$, and $\mathbf{A}^T\mathbf{B} = \mathbf{F}_{\Ab}^T(0)\mathbf{F}_\Bb(0)$. However, $\deg(\mathbf{F}_{\Ab}^T(x)\mathbf{F}_\Bb(x)) = 2t-2 \neq t-1$. In addition, the coefficients of $\mathbf{F}_{\Ab}^T(x)\mathbf{F}_\Bb(x)$ do not have the distribution that we wish in Shamir secret sharing.

	In the BGW scheme, to have Shamir shares of $\mathbf{A}^T\mathbf{B}$ at each node, we use the following approach.
	
	We note that from Lagrange interpolation rule \cite{bakhvalov1977numerical}, if $N \geq 2t-1$,  there exists a vector $\mathbf{r} = (r_1,r_2,\dots,r_N) \in \mathbb{F}^N$ such that
	\begin{align}\label{tarkibe khati}
		\mathbf{H}(0) = \mathbf{A}^T\mathbf{B} = \sum_{n=1}^{N}r_n\mathbf{H}(\alpha_n).
	\end{align}
	
	In this approach, worker $n$ shares $\mathbf{H}(\alpha_n)=\mathbf{F}^T_\Ab(\alpha_n)\mathbf{F}_\Bb(\alpha_n)$ with other workers using Shamir secret sharing. In other words, worker $n$ forms a polynomial of degree $t-1$, 
	\begin{align}
		\mathbf{F}_n(x) = \mathbf{H}(\alpha_n) + \mathbf{\bar{H}}_1^{(n)}x + \mathbf{\bar{H}}_2^{(n)}x^2 + \dots + \mathbf{\bar{H}}_{t-1}^{(n)}x^{t-1}, \nonumber
	\end{align}
	where $\mathbf{\bar{H}}_i^{(n)}$, $i \in [t-1]$, $n \in [N]$,  are chosen independently and uniformly at random in $\mathbb{F} ^{m \times m}$. Worker $n$ sends  the value of $\mathbf{F}_n(\alpha_{n ^{\prime}})$ to  worker ${n ^{\prime}}$, for all $n, n' \in [N]$. 
Then each worker $n$ calculates $\sum_{n'=1}^{N} r_{n'} \mathbf{F}_{n'} (\alpha_{n})$. 
We claim that the result is indeed Shamir sharing of $\mathbf{A}^T\mathbf{B}$. To verify that, let us define $\mathbf{F} (x)$ as
\begin{align}
\mathbf{F} (x) &\triangleq \sum_{n'=1}^{N} r_{n'} \mathbf{F}_{n'} (x)  \nonumber \\
& = \sum_{n'=1}^{N} r_{n'}  \mathbf{H}(\alpha_{n'}) + x\sum_{n'=1}^{N} r_{n'}  \mathbf{\bar{H}}_1^{(n')} + x^2\sum_{n'=1}^{N} r_{n'}  \mathbf{\bar{H}}_2^{(n')} + \dots + x^{t-1}   \sum_{n'=1}^{N} r_{n'}  \mathbf{\bar{H}}_{t-1}^{(n')}.
\end{align} 
Then, we have the following observations:
\begin{enumerate}[(i)]
\item Due to \eqref{tarkibe khati}, $\mathbf{F} (0)=  \sum_{n'=1}^{N} r_{n'}  \mathbf{H}(\alpha_{n'})=  \mathbf{A}^T\mathbf{B}$. 

\item   Worker $n$ has access to $\mathbf{F} (\alpha_n)= \sum_{n'=1}^{N} r_{n'} \mathbf{F}_{n'} (\alpha_{n})$.

\item $\sum_{n'=1}^{N} r_{n'}  \mathbf{\bar{H}}_i^{(n')}$, $i \in [t-1]$,  are independent with uniform distribution in $\mathbb{F}^ {m \times m}$.
\end{enumerate}

Thus, $\mathbf{F} (\alpha_n)$ is indeed a Shamir share of $\mathbf{A}^T\mathbf{B}$. In addition,  if $t$ of the workers send  their samples of $\mathbf{F}(x)$ to the master, it can recover $\mathbf{F} (0)=  \mathbf{A}^T\mathbf{B}$.  

\item	\emph{Phase 3 - Reconstruction:}
	
Each worker $n$ sends  $\mathbf{F} (\alpha_n)$ to the master. It can recover $\mathbf{F} (0)=  \mathbf{A}^T\mathbf{B}$, if it has $\mathbf{F} (\alpha_n)$ from $t$ workers. 

\end{enumerate}

We note that the master can also recover $\sum_{n'=1}^{N} r_{n'}  \mathbf{\bar{H}}_i^{(n')}$, for $i \in [t-1]$, which reveal no information about 
 $\mathbf{A}$ and $\mathbf{B}$. Therefore,  the privacy at the master is guaranteed. In addition, whatever is shared with a worker is based on Shamir secret sharing with new random coefficients. This can be used to prove that the privacy at the workers is guaranteed. 

This example shows that if the number of workers is $N \geq 2t-1$, the system can calculate multiplication.

\end{example}

\begin{remark} \label{remarkfor Shamir}
		In this remark, we explain why we prefer to have Shamir shares of $\mathbf{A}^T\mathbf{B}$ at the workers. The main reason is that this approach gives us great flexibility to use it iteratively and calculate any polynomial function of the inputs. Let us assume that there are three inputs $\mathbf{A}, \mathbf{B}$, and $\mathbf{C}$, and the goal is to calculate $\mathbf{C}^T\mathbf{A}^T\mathbf{B}$. First, we use the above approach to have the Shamir shares of $\mathbf{D} = \mathbf{A}^T\mathbf{B}$ at each worker and use it again to calculate $\mathbf{C}^T\mathbf{D}$. Similarly, if the goal is to calculate $\mathbf{A}^T\mathbf{B}+\mathbf{C}$, we use the above scheme to have the Shamir shares of $\mathbf{D} = \mathbf{A}^T\mathbf{B}$ at each node, and use the scheme of Example \ref{avalin mesal} to calculate $\mathbf{D}+\mathbf{C}$. 
		
		Another important reason is that to calculate $\mathbf{C}^T\mathbf{A}^T\mathbf{B}$, still we need $N = 2t-1$ workers. Otherwise, $\mathbf{F}_{\Cb}^T(x)\mathbf{F}_{\Ab}^T(x)\mathbf{F}_\Bb(x)$ will have degree $3(t-1)$, and thus we will need $N = 3t-2$ workers. In other words, using this approach,  
		the number of servers needed does not grow with the number of matrix multiplication.
		
		This approach also makes the proof of privacy simpler.
	\end{remark}
		
Now let us consider the scenario where there is a storage limit for each worker, i.e., $k > 1$. One approach to deal with this case is to split the job into smaller jobs and use the BGW scheme for each sub-job, as we see in the following example.

\begin{example}[\emph{Concatenation of Job Splitting and the BGW}, $k=2$]
	\label{eg:problem USDMPC1}
	Here, we revisit Example~\ref{mesale2}, but here, we assume that $k=2$.

 We partition each matrix into two sub-matrices as follows: 
		\begin{align}
		\label{eq:matrixpartition1}
	\mathbf{A}=\begin{bmatrix} \mathbf{A}_1 & \mathbf{A}_2 \end{bmatrix}, \\
	\label{eq:matrixpartition2}
	\mathbf{B} =\begin{bmatrix} \mathbf{B}_1 & \mathbf{B}_2 \end{bmatrix}, 	
	\end{align}
	where $\mathbf{A}_i,\mathbf{B}_i \in \mathbb{F}^{m \times \frac{m}{2}}$, for $i\in\{ 1,2 \}$. We note that
	\begin{align}
	\mathbf{A}^T\mathbf{B} = \begin{bmatrix} \mathbf{A}_1^T\mathbf{B}_1 & \mathbf{A}_1^T\mathbf{B}_2 \\
	\mathbf{A}_2^T\mathbf{B}_1 & \mathbf{A}_2^T\mathbf{B}_2 
	\end{bmatrix}. \nonumber		
	\end{align}
	
Therefore, to calculate $\mathbf{A}^T\mathbf{B}$, we  can use the BGW scheme with four groups of workers to calculate $\mathbf{A}^T_i\mathbf{B}_j, i,j \in \{ 1,2 \}$, each of them with at least $2t-1$ workers, following Example \ref{mesale2}. Therefore, using this scheme,  we need $N = 4(2t-1)$ workers.

\end{example}

One can see that with the concatenation of job-splitting and the BGW scheme described in Example \ref{eg:problem USDMPC1}, the minimum number of workers needed to calculate the addition and multiplication of two matrices is
$N = kt$ and $N =k^2 (2t-1)$, respectively. In this paper, we propose an algorithm which reduces the required number of workers significantly.	

\section{Main Results}\label{main result}
The main result of this paper is as follows:  

\begin{theorem}\label{remark}
	For any $k, t \in  \mathbb{N}$ and any polynomial function $\mathbf{G}$,
		\begin{align}
	N_{\mathbf{G}}^\star(t,k) \leq \min\{2k^2 + 2t -3, k^2 + kt + t -2\}, \nonumber
	\end{align}
	where 
	\[   
	\min \{2k^2 + 2t - 3, k^2 + kt + t -2\} = 
	\begin{cases}
	2k^2 + 2t - 3 	&\quad\text{if} ~k < t\\
	k^2 + kt + t -2 &\quad\text{if} ~k \geq t
	\end{cases}.
	\]

\end{theorem}

\begin{remark} 
To prove Theorem \ref{remark}, we propose a scheme which is based on a novel approach for sharing called \emph{polynomial sharing} and some procedures to calculate basic functions such as addition and multiplication. Using these procedures iteratively, we can calculate any polynomial function.
\end{remark}

\begin{remark}  
Recall that concatenation of job splitting and MPC for multiplication needs $k^2(2t-1)$ workers. However in the proposed scheme we can do multiplication with at most $\min\{2k^2 + 2t -3, k^2 + kt + t -2\}$ workers, which is an orderwise improvement. For example for $t = 200$ and $k = 16$, the proposed scheme needs $N = 909$ workers, while the job splitting approach needs $N = 102144$ workers. 
\end{remark} 

\begin{remark}
		Theorem \ref{remark} is about the cases where there is at least one matrix multiplication in the calculation of the function $\mathbf{G}$. If $\mathbf{G}$ is a linear function of the inputs, then the proposed scheme needs $k+t-1$ workers.
	\end{remark}
	
	\begin{remark}
		As mentioned in the introduction, secure matrix multiplication, which has been introduced in~\cite{chang2018capacity},  concurrently by the conference paper of this manuscript~\cite{Multi-Party}, focuses on calculating the multiplication of two matrices. Ignoring the privacy constraint at the master, secure matrix multiplication can be considered as a special case of our proposed problem formulation. This line of work has been followed by \cite{KES19, d2020gasp} to improve its efficiency. 
		The number of workers needed by the scheme of \cite{chang2018capacity} is equal to $(k+t-1)^2$, which is outperformed by our proposed scheme~\eqref{remark}. The number of workers needed by the follow-up paper~\cite{KES19} is equal to $k^2+tk+t-2$,  which is again outperformed by the proposed scheme~\eqref{remark}.
		Reference \cite{d2020gasp} reports two schemes, named as  GAPS-Big and GAPS-Small.
		Indeed, GAPS-Big is the same as what we had already reported in~\cite{Multi-Party}.  The number of servers needed by GAPS-Small is smaller than what we report in this paper, for the case where $3 \leq t < k$. The results are summarized in Table~\ref{table:N}.
	\end{remark}
	
	\begin{table}[ht]
		\caption{The number of workers needed for Matrix Multiplication (SMM)}
		\centering
		\begin{tabular}{l | l | l }
			\hline\hline
			Work & Date & Number of Workers Needed \\ [0.5ex] 
			\hline
			Our Scheme (Nodehi-MaddahAli)~\cite{Multi-Party}  & June 2018 &   
			$N= 	
			\begin{cases}
			2k^2 + 2t - 3 	&\quad\text{if} \  k < t\\
			k^2 + kt + t -2 &\quad\text{if} \ k \geq t
			\end{cases}$ \\
			\hline
			SMM (Chang-Tandon~\cite{chang2018capacity}) &   June 2018 & $N= (k+t-1)^2$ 
			\\
			\hline
			SMM (Kakar et al.~\cite{KES19})  &   Oct. 2018 & $N=k^2+tk+t-2$
			\\
			\hline 
			GASP-Big (D' Oliveira et. al.  ~\cite{d2020gasp})  &   Dec. 2018 & $N= 	
			\begin{cases}
			2k^2 + 2t - 3 	&\quad\text{if} \ k < t\\
			k^2 + kt + t -2 &\quad\text{if}\ k \geq t
			\end{cases}$ 
			\\
			GASP-Small (D' Oliveira et. al.~\cite{d2020gasp})  &   Dec. 2018 & $N= 	
			\begin{cases}
			k^2 + 2k 	&\quad\text{if} \ 2 = t \leq k\\
			k^2 + 2k + (t-1)^2 +t-4 &\quad\text{if}\ 3 \leq t \leq k \\
			k^2 + kt + 2t-5-\lfloor{\frac{t-3}{k}}{\rfloor} 	&\quad\text{if} \ k < t \leq k(k-1)+2\\
			2k^2 + kt + t -2k-1 &\quad\text{if}\ k(k-1)+2 < t 
			\end{cases}$ 
			\\
			\hline
		\end{tabular}
		\label{table:N}
	\end{table}

\section{Motivating example} \label{Motivating Examples for proposed scheme}
 Here, we revisit Example~\ref{eg:problem USDMPC1} with $k=2$ and $t=4$ and propose a solution that needs 13 workers, as compared to $4 \times 7 = 28 $ workers in the solution of Example~\ref{eg:problem USDMPC1}.

 \begin{example}[\emph{The proposed Approach for $k = 2$ and $t = 4$}]
 	\label{eg:problem CSDMPC1}
 	In this example we explain the proposed scheme to securely calculate $\mathbf{L} = \mathbf{A}^T\mathbf{B}$. 
	\begin{enumerate}
		\item \emph{Phase 1 - Sharing:}

 	Consider the following \emph{polynomial functions}
 	\begin{align}
 	\mathbf{F}_A(x) &= \mathbf{A}_1 + \mathbf{A}_2x + 
 	\mathbf{\bar{A}}_3x^4 + 
 	\mathbf{\bar{A}}_4x^5 + 
 	\mathbf{\bar{A}}_5x^6, \nonumber\\
 	\mathbf{F}_{\mathbf{B}}(x) & = \mathbf{B}_1 + \mathbf{B}_2x^2 + 
 	\mathbf{\bar{B}}_3x^4 + \mathbf{\bar{B}}_4x^5 + \mathbf{\bar{B}}_5x^6, \nonumber
 	\end{align}
 	where $\mathbf{A}_1, \mathbf{A}_2, \mathbf{B}_1$, and  $ \mathbf{B}_2$ are defined in \eqref{eq:matrixpartition1} and \eqref{eq:matrixpartition2} in Example \ref{eg:problem USDMPC1}. In addition, $\mathbf{\bar{A}}_i, \mathbf{\bar{B}}_i$, for $i\in[3,5]$, are  matrices chosen independently and uniformly at random in $\mathbb{F}^{m \times \frac{m}{2}}$. These polynomials follow a certain pattern. More precisely, in these polynomials the coefficients of some powers of $x$ are zero. In $\mathbf{F}_{\mathbf{A}}(x)$ the coefficients of $x^2$ and $x^3$ are zero and in $\mathbf{F}_{\mathbf{B}}(x)$ the coefficients of $x$ and $x^3$ are zero. We choose $\alpha_1, \alpha_2, \dots, \alpha_N \in \mathbb{F}$, independently and uniformly at random. Source nodes $1$ and $2$ share $\mathbf{F}_{\mathbf{A}}(\alpha_n)$ and $\mathbf{F}_{\mathbf{B}}(\alpha_n)$ with worker $n$, respectively.
	We call this form of sharing as \emph{polynomial sharing}. 
 
 		\item \emph{Phase 2 - Computation and Communication:}
		
 	Worker $n$ calculates $\mathbf{F}_{A}^T(\alpha_n)\mathbf{F}_{\mathbf{B}}(\alpha_n)$. Consider the polynomial function $\mathbf{H}(x)$ of degree $12$, defined as 
 	\begin{align}\label{zaribe h}
 	\mathbf{H}(x) = \sum_{n=0}^{12} \mathbf{H_n}x^n \defeq \mathbf{F}_{A}^T(x)\mathbf{F}_{\mathbf{B}}(x).
 	\end{align}
	We note that
	\begin{align}\label{coefficients of h}	
 	\mathbf{H}_0 =\mathbf{A}_{1}^T\mathbf{B}_1,\nonumber\\
 	\mathbf{H}_1 =\mathbf{A}_{2}^T\mathbf{B}_1, \nonumber\\
 	\mathbf{H}_2 =\mathbf{A}_{1}^T\mathbf{B}_2,\nonumber\\
 	\mathbf{H}_3 =\mathbf{A}_{2}^T\mathbf{B}_2. 
 	\end{align}
If the master has $\mathbf{H}(\alpha_n)$ for $N \geq 13$ distinct $\alpha_i$'s, then it can calculate all the coefficients of $\mathbf{H}(x)$, including $\mathbf{H}_0 =\mathbf{A}_{1}^T\mathbf{B}_1$,
  	$\mathbf{H}_1 =\mathbf{A}_{2}^T\mathbf{B}_1$, 
 	$\mathbf{H}_2 =\mathbf{A}_{1}^T\mathbf{B}_2$, and
	 $\mathbf{H}_3 =\mathbf{A}_{2}^T\mathbf{B}_2$,  with probability approaching to one, as $|\mathbb{F}| \rightarrow \infty$. More precisely, according to Lagrange interpolation rule \cite{bakhvalov1977numerical}, there are some  $r^{(i,j)}_n$, $i,j \in \{1,2\}$ and $n\in [N]$, such that
	\begin{align}
	\label{eq:r}
	\mathbf{A}_{i}^T\mathbf{B}_j= \sum_{n=1}^N r^{(i,j)}_n\mathbf{H}(\alpha_n). 
	 \end{align}
	 \begin{remark}
	 Note that $r^{(i,j)}_n$, $i,j \in \{1,2\}$ and $n\in [N]$ are only functions of $\alpha_n$, $ n \in [N]$, which are known by all workers. 
	\end{remark}
 	
Similar to Example \ref{mesale2} and following Remark \ref{remarkfor Shamir}, in order to prepare the scene for the next stage and cast the result of the local computation in the form of the \emph{polynomial sharing}, we use the following steps. Worker $n$ forms  $\mathbf{Q}^{(n)}(x)$, defined as 
 	\begin{align}\label{tarife q koochik}
 	\mathbf{Q}^{(n)}(x) &\defeq 
 	\begin{bmatrix}
 	r^{(1,1)}_n\mathbf{H}(\alpha_n)\\
 	r^{(2,1)}_n\mathbf{H}(\alpha_n)
 	\end{bmatrix} + 
 	\begin{bmatrix}
 	r^{(1,2)}_n\mathbf{H}(\alpha_n)\\
 	r^{(2,2)}_n\mathbf{H}(\alpha_n)
 	\end{bmatrix}x^2  \nonumber \\
 	&+ \mathbf{R}^{(n)}_{0}x^4  + \mathbf{R}^{(n)}_{1}x^5 + \mathbf{R}^{(n)}_{2}x^6,
 	\end{align}
 	where $\mathbf{R}^{(n)}_{i}, i\in \{0,1,2\}$, are chosen independently and uniformly at random in $\mathbb{F}^{m \times \frac{m}{2}}$. Recall that $\mathbf{H}(\alpha_n), r^{(1,1)}_n, r^{(1,2)}_n, r^{(2,1)}_n ~\text{and}~ r^{(2,2)}_n$ are available at worker $n$. Thus, worker $n$ has all the information to form $\mathbf{Q}^{(n)}(x)$. Then worker $n$ sends 	$\mathbf{Q}^{(n)}(\alpha_{n'})$ to worker ${n'}$, for all $n' \in [N]$.

Thus, worker $n'$ will have access to the matrices $\{\mathbf{Q}^{(1)}(\alpha_{n'}), \mathbf{Q}^{(2)}(\alpha_{n'}), \dots, \mathbf{Q}^{(N)}(\alpha_{n'})\}$. Then worker $n'$ calculates $\sum_{n=1}^{13}\mathbf{Q}^{(n)}(\alpha_{n'})$.  We claim that this summation is indeed the polynomial share of the matrix $\mathbf{A}^T\mathbf{B}$. To verify that, consider the  polynomial function
 	\begin{align}\label{tarife Q bozorg}
 	\mathbf{Q}(x) \defeq 
 	\sum_{n=1}^{N}\mathbf{Q}^{(n)}(x).
 	\end{align}	

	We note that 
	\begin{align}
 	\mathbf{Q}(x) & \defeq \nonumber
 	\sum_{n=1}^{N}\mathbf{Q}^{(n)}(x) \\ \nonumber
	& =
	\sum_{n=1}^{N}
	\begin{bmatrix}
 	r^{(1,1)}_n\mathbf{H}(\alpha_n)\\
 	r^{(2,1)}_n\mathbf{H}(\alpha_n)
 	\end{bmatrix} + 
	x^2 \sum_{n=1}^{N}
 	\begin{bmatrix}
 	r^{(1,2)}_n\mathbf{H}(\alpha_n)\\
 	r^{(2,2)}_n\mathbf{H}(\alpha_n)
 	\end{bmatrix} 
 	+ x^4 \sum_{n=1}^{N} \mathbf{R}^{(n)}_{0} + x^5 \sum_{n=1}^{N}\mathbf{R}^{(n)}_{1} + x^6 \sum_{n=1}^{N}\mathbf{R}^{(n)}_{2} \\
	\nonumber
	& \overset{(a)} {=}
	\begin{bmatrix}
 	\mathbf{A}_{1}^T\mathbf{B}_1\\
 	\mathbf{A}_{2}^T\mathbf{B}_1
 	\end{bmatrix} + 
	x^2 
 	\begin{bmatrix}
 	\mathbf{A}_{1}^T\mathbf{B}_2\\
 	\mathbf{A}_{2}^T\mathbf{B}_2
 	\end{bmatrix} 
 	+ x^4 \sum_{n=1}^{N} \mathbf{R}^{(n)}_{0} + x^5 \sum_{n=1}^{N}\mathbf{R}^{(n)}_{1} + x^6 \sum_{n=1}^{N}\mathbf{R}^{(n)}_{2} \\
	& =\mathbf{L}_1 + \mathbf{L}_2x^2 + \mathbf{\bar{L}}_3x^4 + \mathbf{\bar{L}}_4x^5 + \mathbf{\bar{L}}_5x^6, 
 	\end{align}
 	where (a) follows from~\eqref{eq:r} and $\mathbf{L}_1 \defeq  \begin{bmatrix} \mathbf{A}_{1}^T\mathbf{B}_1 \\\mathbf{A}_{2}^T\mathbf{B}_1\end{bmatrix}, \mathbf{L}_2 \defeq \begin{bmatrix} \mathbf{A}_{1}^T\mathbf{B}_2 \\\mathbf{A}_{2}^T\mathbf{B}_2\end{bmatrix}$, 
	$\mathbf{\bar{L}}_3\defeq \sum_{n=1}^{N} \mathbf{R}^{(n)}_{0}$,   $\mathbf{\bar{L}}_4\defeq \sum_{n=1}^{N}\mathbf{R}^{(n)}_{1}$, and  $\mathbf{\bar{L}}_5\defeq \sum_{n=1}^{N}\mathbf{R}^{(n)}_{2}$. Note that,  $\mathbf{\bar{L}}_i, i \in \{3,4,5\}$, have independent and  uniform distribution in $\mathbb{F}^{m \times \frac{m}{2}}$.
 	Because of the randomness of $\mathbf{R}^{(n)}_{i}(x)$, for $i \in \{0,1,2\}$, $n \in [N]$, and degree of $\mathbf{Q}(x)$, one can easily verify that $\mathbf{Q}(x)$ is in the form of polynomial sharing and worker $n$ has access to $\mathbf{Q}(\alpha_n)$.

 		\item	\emph{Phase3 - Reconstruction:}

 	Worker $n$ sends  $\mathbf{Q}(\alpha_{n})$ to the master, $n \in [N]$. The master then calculates $\mathbf{A}^T\mathbf{B}$ from what it receives by polynomial interpolation. 
	\end{enumerate}
		
 \end{example}
 
This algorithm guarantees conditions \eqref{privacy for workers}, the privacy at the workers. An intuitive explanation is that in each interaction among the workers, the algorithm  uses some independent random matrices, as the coefficients of the sharing polynomials,  such that any subset,  including $3$ workers, cannot gain any information about the input data. This will be proven formally later in Appendix \ref{Security analysis} for the general cases. Also there is no information leakage at the master, as required according to the privacy constraints \eqref{privacy for master}, which will be formally proven in Appendix \ref{Security analysis}.

	\begin{example}
		\label{example:AB+C}
		In this example, we aim to show how to use a new approach to securely calculate $\mathbf{M} = \mathbf{A}^T\mathbf{B} + \mathbf{C}$, for $k = 2$ and $t = 4$.
		\begin{enumerate}
			\item \emph{Phase 1 - Sharing:}
			
			Consider the following \emph{polynomial functions}
			\begin{align}
			\mathbf{F}_{\mathbf{A}}(x) &= \mathbf{A}_1 + \mathbf{A}_2x + \mathbf{\bar{A}}_3x^4 + \mathbf{\bar{A}}_4x^5 + \mathbf{\bar{A}}_5x^6, \nonumber\\
			\mathbf{F}_{\mathbf{B}}(x) & = \mathbf{B}_1 + \mathbf{B}_2x^2 + \mathbf{\bar{B}}_3x^4 + \mathbf{B}_4x^5 + \mathbf{\bar{B}}_5x^6, \nonumber \\
			\mathbf{F}_{\mathbf{C}}(x) & = \mathbf{C}_1 + \mathbf{C}_2x^2 + \mathbf{\bar{C}}_3x^4 + \mathbf{\bar{C}}_4x^5 + \mathbf{\bar{C}}_5x^6, \nonumber
			\end{align}
			where in the above equations, $\mathbf{A}_1, \mathbf{A}_2, \mathbf{B}_1, \mathbf{B}_2, \mathbf{C}_1$, and $\mathbf{C}_2$ are defined similar to \eqref{eq:matrixpartition1}. In addition, $\mathbf{\bar{A}}_i, \mathbf{\bar{B}}_i, \mathbf{\bar{C}}_i$, for $i\in[3,5]$, are  matrices chosen independently and uniformly at random in $\mathbb{F}^{m \times \frac{m}{2}}$.
			Similar to Example \ref{eg:problem CSDMPC1}, source nodes $1$, $2$, and $3$ share $\mathbf{F}_{\mathbf{A}}(\alpha_n)$, $\mathbf{F}_{\mathbf{B}}(\alpha_n)$, and $\mathbf{F}_{\mathbf{C}}(\alpha_n)$ with worker $n$.
			
			\item \emph{Phase 2 - Computation and Communication:}
			In this phase, first, all of the workers follow phase 2 of Example \ref{eg:problem CSDMPC1}, in order to have access to $\mathbf{Q}(\alpha_n)$ defined in \eqref{tarife Q bozorg}, as the polynomial shares of $\mathbf{A}^T\mathbf{B}$. Then each worker calculates $\mathbf{O}_n = \mathbf{Q}(\alpha_n) + \mathbf{F}_{\mathbf{C}}(\alpha_n)$. 
			
			\item	\emph{Phase3 - Reconstruction:}
			
			Worker $n$  sends  $\mathbf{O}_n$ to the master, $n \in [N]$. The master then calculates $\mathbf{A}^T\mathbf{B} + \mathbf{C}$ from what it receives by polynomial interpolation.
			
		\end{enumerate}
		
	\end{example}


\begin{remark}
		We would like to emphasize that recasting the result of each round of computing in the form of a polynomial-sharing  is a very crucial step for us and has several important advantages:  
		\begin{enumerate}
			\item This approach gives us great flexibility to use the procedures iteratively and calculate any polynomial function of the inputs. Let us assume that there are three inputs $\mathbf{A}, \mathbf{B}$, and $\mathbf{C}$, and the goal is to calculate $\mathbf{C}^T\mathbf{A}^T\mathbf{B}$. First, we use the multiplication procedure to have the shares of $\mathbf{D} = \mathbf{A}^T\mathbf{B}$. At the end of the multiplication procedure, the workers have the polynomial shares of $\mathbf{D}$ (instead of the multiplications of the shares of $\mathbf{A}$ and $\mathbf{B}$). The workers are also given the polynomial shares of  $\mathbf{C}$.  Thus we can use the multiplication procedure again to calculate $\mathbf{C}^T\mathbf{D}$. Similarly, if the goal is to calculate $\mathbf{A}^T\mathbf{B}+\mathbf{C}$ (see Example~\ref{example:AB+C}), we use the multiplication procedure to have the polynomial shares of $\mathbf{D} = \mathbf{A}^T\mathbf{B}$ at each node, and use the addition procedure to calculate $\mathbf{D}+\mathbf{C}$.  
			
			\item It allows us to develop a scheme, where the number of workers needed does NOT grow with the degree of the polynomial function $\mathbf{G}(.)$. Again assume that the goal is to calculate  $\mathbf{C}^T\mathbf{A}^T\mathbf{B}$. Reviewing the argument above, one can confirm that the number of workers needed to calculate $\mathbf{C}^T\mathbf{A}^T\mathbf{B}$ is the same as the number of workers needed to calculate $\mathbf{A}^T\mathbf{B}$.

			\item It simplifies the proof of privacy, considering the fact that we may have several rounds of computation.  Reviewing the proof of privacy in Appendix \ref{Security analysis} of the paper, we observe that recasting the result of each round in the form of polynomial sharing helps us develop an iterative argument and establish that the information leakage is zero. 
			Without that, it would be very hard to prove the privacy, considering all the interaction among the workers to keep the number of workers constant. 
		\end{enumerate}
	\end{remark}

\section{Polynomial Sharing}
\label{Polynomial Sharing}
In this section, we formally define a sharing scheme, called \emph{polynomial sharing}, and in the next section, we show that it admits basic operations such as addition, multiplication of matrices, and transposing a matrix. By concatenating the procedures for basic operations, we show that any polynomial function of the input data can be calculated,  subject to the problem constraints. This scheme is motivated by \cite{yu2017polynomial}, which is a coding technique for matrix multiplication in the distributed system with stragglers.
\begin{definition}
		Let $\mathbf{A} \in \mathbb{F}^{m \times m}$, partitioned as 
		\begin{align}\label{tarifeA}
		\mathbf{A} \defeq \begin{bmatrix}
		\mathbf{A}_{1}, & \mathbf{A}_{2}, & \dots, & \mathbf{A}_{k}
		\end{bmatrix},
		\end{align}
		where $\mathbf{A}_j \in \mathbb{F}^{m \times \frac{m}{k}} $, for some $k \in \mathbb{N}$, $k|m$, and $j \in [k]$. The polynomial function $\mathbf{F}_{\Ab,b,t,k}(x)$, for some $b \in [k]$, is defined as
		\begin{align}\label{tarife raveshe sharing}
		\mathbf{F}_{\Ab,b,t,k}(x) \defeq \sum_{j=1}^{k}\mathbf{A}_jx^{b(j-1)} + \sum_{j=1}^{t-1} \mathbf{R}_{n}x^{k^2+j-1},
		\end{align}
		where $\mathbf{A}_j, j = 1, 2, \dots, k$, are defined in (\ref{tarifeA}) and  $\mathbf{R}_j, j= 1, 2, \dots, t-1$, are chosen independently and uniformly at random from $\mathbb{F}^{m \times \frac{m}{k}}$. 
		We say that matrix $\Ab$ is $(b,t,k)$ polynomial-shared with workers in $[N]$, if $\mathbf{F}_{\Ab,b,t,k}(\alpha_n)$ is sent to worker $n$, where $\alpha_n \in \F$ are distinct constants assigned to worker $n$, $n\in [N]$, where $k+t-1 < N$. 
	\end{definition}

\begin{remark}
	Note that for $b = 1$ and $k = 1$, the polynomial sharing is reduced to Shamir secret sharing.
\end{remark}

The following theorem about polynomial sharing will be essential in the following sections.

\begin{theorem} \label{qazie sharing}
	Let $t,k,m \in \mathbb{N}$, $N \geq \min\{ 2k^2 + 2t -3, k^2 + tk + t -2\}$, $k |m$, and $\mathbf{A},\mathbf{B} \in \mathbb{F} ^ {m \times m}$. Define 
	\begin{align*}
		\mathbf{H}(x) \defeq \mathbf{F}_{\Ab,1,t,k}^T(x)\mathbf{F}_{\Bb,k,t,k}(x).
	\end{align*}
	Then for some large enough  $|\mathbb{F}|$, there exist distinct $\alpha_1, \alpha_2, \dots, \alpha_N \in \mathbb{F}$, such that for any distinct $i_1, i_2, \dots, i_{k+t-1} \in [N]$, we have
	\begin{align}
	H(\mathbf{A}|\mathbf{F}_{\Ab,1,t,k}(\alpha_{i_1}), \mathbf{F}_{\Ab,1,t,k}(\alpha_{i_2}), \dots, \mathbf{F}_{\Ab,1,t,k}(\alpha_{i_{k+t-1}})) 
	& = 0, \label{sharte1}\\
	H(\mathbf{B}|\mathbf{F}_{\Bb,k,t,k}(\alpha_{i_1}), \mathbf{F}_{\Bb,k,t,k}(\alpha_{i_2}), \dots, \mathbf{F}_{\Bb,k,t,k}(\alpha_{i_{k+t-1}})) 
	& = 0, \label{sharte2}
	\end{align}
	and
	\begin{align}	
	H(\mathbf{A}^T\mathbf{B}|\mathbf{H}(\alpha_1), \mathbf{H}(\alpha_2), \dots, \mathbf{H}(\alpha_N)) = 0 \label{sharte3}.
	\end{align}
	In addition, if we choose $\alpha_1, \alpha_2, \dots, \alpha_N$, independently and uniformly at random in $\mathbb{F}$, the probability that (\ref{sharte1}), (\ref{sharte2}), and (\ref{sharte3}) hold, approaches to one, as $|\mathbb{F}| \rightarrow \infty$.
\end{theorem}
\begin{proof}
	See Appendix \ref{Appendice A}.
\end{proof}

\section{Procedures}
\label{Procedures}
In this section, we explain several procedures to do basic operations such as addition, multiplication, and transposing, using polynomial sharing without leaking any information. By concatenating these procedures, we can calculate any polynomial function of the private input data, subject to the constraints \eqref{correctness}, \eqref{privacy for workers}, and \eqref{privacy for master}.
\subsection{Addition}
	Let $\mathbf{A}, \mathbf{B} \in \mathbb{F}^ {m \times m}$, and 
	\begin{align}
		\mathbf{A}=&
		\begin{bmatrix} \mathbf{A}_{1}, & \mathbf{A}_{2}, & \dots, & \mathbf{A}_{k}
		\end{bmatrix}, \nonumber \\ 
		\mathbf{B}=& 
		\begin{bmatrix} \mathbf{B}_{1}, & \mathbf{B}_{2}, & \dots, & \mathbf{B}_{k}
		\end{bmatrix}	\nonumber ,
	\end{align}
	where $\mathbf{A}_i, \mathbf{B}_i \in \mathbb{F}^{m \times \frac{m}{k}}$, for $i \in [k]$ and $k|m$. In addition, assume $\Ab$ and $\Bb$ are $(b,t,k)$ polynomial-shared with workers in $[N]$.
	The objective is to $(b,t,k)$ polynomial-share $\Lb \defeq \Ab + \Bb$ with workers in $[N]$. To do that we follow a one-step procedure, as follows:
	
	\begin{enumerate}
		\item \emph{Step 1- Computation:}
		
		In this step, worker $n$ calculates $\mathbf{F}_{\Ab,b,t,k}(\alpha_{n}) +‌\mathbf{F}_{\Bb,b,t,k}(\alpha_{n})$.
		We claim that this summation is indeed the polynomial share of the matrix $\mathbf{A} + \mathbf{B}$. To verify that, consider the  polynomial function
			\begin{align}\label{tarife Q1 bozorg1}
		\mathbf{Q}(x) \defeq 
		\mathbf{F}_{\Ab,b,t,k}(x) +‌\mathbf{F}_{\Bb,b,t,k}(x).
		\end{align}
		 We note that
		\begin{align}
		\mathbf{Q}(x) & = \nonumber
		\mathbf{F}_{\Ab,b,t,k}(x) +‌\mathbf{F}_{\Bb,b,t,k}(x) \\ \nonumber
		& = \sum_{j=1}^{k}\mathbf{A}_jx^{b(j-1)} + \sum_{j=1}^{t-1} \mathbf{\bar{R}}_{j}x^{k^2+j-1} + 
		\sum_{j=1}^{k}\mathbf{B}_jx^{b(j-1)} + \sum_{j=1}^{t-1} \mathbf{\hat{R}}_{j}x^{k^2+j-1} \\ \nonumber
		& = \sum_{j=1}^{k} (\Ab_j + \Bb_j)x^{b(j-1)} +‌\sum_{j=1}^{t-1} (\mathbf{\bar{R}}_{j} + \mathbf{\hat{R}}_{j})x^{k^2+j-1} \\ 
		& \overset{(a)} {=}
		\sum_{j=1}^{k}\mathbf{L}_jx^{b(j-1)} + \sum_{j=1}^{t-1} \mathbf{R}_{j}x^{k^2+j-1}, 
		\end{align}
		where (a) follows from the definitions $\mathbf{L}_j \defeq \mathbf{A}_j + \mathbf{B}_j $ for $j \in [k]$, and  $\mathbf{R}_{j} \defeq \mathbf{\bar{R}}_{j} + \mathbf{\hat{R}}_{j} $, for $j \in [t-1]$. We note that $\Lb = \begin{bmatrix}
		\Lb_1 & \Lb_2 & \dots & \Lb_k
		\end{bmatrix}$. In addition, $\mathbf{R}_{j}, j \in [t-1]$, have independent and  uniform distribution in $\mathbb{F}^{m \times \frac{m}{k}}$. Thus, $\mathbf{Q}(x)$ is in the form of $(b,t,k)$ polynomial sharing of $\Ab + \Bb$ with workers in $[N]$, and worker $n$ has access to $\mathbf{Q}(\alpha_n)$.  This step is detailed in Algorithm~\ref{Addition}. 
	\end{enumerate}
	\begin{algorithm}[!htbp]
		\caption{Addition}
		\label{Addition}
		\begin{algorithmic}[1]
			\Statex
			\textbf{Inputs:}
			Matrices $\mathbf{A}, \mathbf{B} \in \mathbb{F}^ {m \times m}$ which are $(b,t,k)$ polynomial-shared with workers in $[N]$. 
			\State
			Worker $n$ adds its shares of $\Ab$ and $\Bb$.
			\State
			End.
		\end{algorithmic}
	\end{algorithm}

\subsection{Multiplication by Constant}
	Let $\mathbf{A} \in \mathbb{F}^ {m \times m}$, and
\begin{align}
\mathbf{A}=&
\begin{bmatrix} \mathbf{A}_{1}, & \mathbf{A}_{2}, & \dots, & \mathbf{A}_{k}
\end{bmatrix}, \nonumber
\end{align}
for $\mathbf{A}_i \in \mathbb{F}^{m \times \frac{m}{k}}$, $i \in [k]$, and $k|m$. In addition, assume $\mathbf{A}$ is $(b,t,k)$ polynomial-shared with workers in $[N]$.
The objective is to $(b,t,k)$ polynomial-share $\Lb \defeq q\mathbf{A}$ with workers in $[N]$, where $q$ is a constant in $\mathbb{F}$. To do that, we follow a one-step procedure, as follows:

\begin{enumerate}
	\item \emph{Step 1- Computation:}
	
	In this step, worker $n$ calculates $q\mathbf{F}_{\Ab,b,t,k}(\alpha_{n})$.
	We claim that this multiplication is indeed the polynomial share of the matrix $q\mathbf{A}$. To verify that, consider the  polynomial function 
	
	\begin{align}\label{tarife Q11 bozorg1}
	\mathbf{Q}(x) \defeq 
	q\mathbf{F}_{\Ab,b,t,k}(x).
	\end{align}
	 We note that
	\begin{align}
	\mathbf{Q}(x) & = \nonumber
	q\mathbf{F}_{\Ab,b,t,k}(x) \\ \nonumber
	& = q\sum_{j=1}^{k}\mathbf{A}_jx^{b(j-1)} + q\sum_{j=1}^{t-1} \mathbf{\bar{R}}_{j}x^{k^2+j-1} \\ \nonumber
	& = \sum_{j=1}^{k} q\Ab_jx^{b(j-1)} +‌\sum_{j=1}^{t-1} q\mathbf{\bar{R}}_{j}x^{k^2+j-1} \\ 
	& \overset{(a)} {=}
	\sum_{j=1}^{k}\mathbf{L}_jx^{b(j-1)} + \sum_{j=1}^{t-1} \mathbf{R}_{j}x^{k^2+j-1}, 
	\end{align}
	where (a) follows from the definitions  $\mathbf{L}_j \defeq q\mathbf{A}_j$, for $j \in [k]$, and  $\mathbf{R}_{j} \defeq q\mathbf{\bar{R}}_{j}$, for $j \in [t-1]$. 
	We note that $\Lb = \begin{bmatrix}
	\Lb_1 & \Lb_2 & \dots & \Lb_k
	\end{bmatrix}$. In addition,  $\mathbf{R}_{j}, j \in [t-1]$, have independent and  uniform distribution in $\mathbb{F}^{m \times \frac{m}{k}}$. Thus, $\mathbf{Q}(x)$ is in the form of $(b,t,k)$ polynomial sharing of $q\mathbf{A}$ with workers in $[N]$, and worker $n$ has access to $\mathbf{Q}(\alpha_n)$.  This step is detailed in Algorithm~\ref{Multiplication by Constant}. 	
\end{enumerate}

 	\begin{algorithm}[!htbp]
 		\caption{Multiplication by Constant}
 		\label{Multiplication by Constant}
 		\begin{algorithmic}[1]
 			\Statex
 			\textbf{Inputs:}
 			Matrix $\mathbf{A} \in \mathbb{F}^ {m \times m}$ which is $(b,t,k)$ polynomial-shared with workers in $[N]$, and $q \in \mathbb{F}$. 
 			\State
 			Worker $n$ multiplies its share of $\Ab$ by $q$.
 			\State
 			End.
 		\end{algorithmic}
 	\end{algorithm} 	
 	\begin{remark}
 		Polynomial sharing scheme has the linearity property. It means that if matrices $\mathbf{A}$ and $\mathbf{B}$ are $(b,t,k)$ polynomial-shared  with workers in $[N]$, in order to $(b,t,k)$ polynomial-share $\mathbf{L} = q\mathbf{A} + p\mathbf{B}$ with workers in $[N]$, it is enough that each worker just locally calculates the same computation on its shares.
 	\end{remark}
\subsection{Multiplication of Two Matrices}
	Similar to the Shamir secrete sharing scheme, calculating the shares of the multiplication of two matrices is not as simple as calculating the addition. Workers need to do some communication. In what follows, we explain the procedure. 
	
	Let us assume that $\mathbf{A}, \mathbf{B} \in \mathbb{F}^ {m \times m}$, and
	\begin{align}
	\mathbf{A}=&
	\begin{bmatrix} \mathbf{A}_{1}, & \mathbf{A}_{2}, & \dots, & \mathbf{A}_{k}
	\end{bmatrix}, \nonumber \\ 
	\mathbf{B}=& 
	\begin{bmatrix} \mathbf{B}_{1}, & \mathbf{B}_{2}, & \dots, & \mathbf{B}_{k}
	\end{bmatrix}	\nonumber ,
	\end{align}
	where $\mathbf{A}_i, \mathbf{B}_i \in \mathbb{F}^{m \times \frac{m}{k}}$, for $i \in [k]$ and $k|m$. In addition, assume $\mathbf{A}$ and $\mathbf{B}$ are $(1,t,k)$ and $(k,t,k)$ polynomial-shared with workers in $[N]$, respectively. The goal is to $(b,t,k)$ polynomial-share $\mathbf{L} \defeq \mathbf{A}^T\mathbf{B}$ with workers in $[N]$. We follow three steps, including computation, communication, and aggregation.
	\begin{enumerate}
		\item \emph{Step 1- Computation:}
		
		Worker $n$ calculates $\mathbf{F}_{\Ab,1,t,k}^T(\alpha_n)\mathbf{F}_{\Bb,k,t,k}(\alpha_n)$. 
		
		Consider the polynomial function $\mathbf{H}(x)$ of degree $2(k^2+t-2)$, defined as, 
		\begin{align} \label{elements of h1}
		\mathbf{H}(x) = \sum_{n=0}^{2(k^2+t-2)} \mathbf{H}_nx^n \defeq \mathbf{F}_{\Ab,1,t,k}^T(x)\mathbf{F}_{\Bb,k,t,k}(x).
		\end{align}
		It is important to note that
		\begin{align}\label{coefficients of h1}
		\mathbf{H}_{i-1+k(j-1)} = \mathbf{A}^T_i\mathbf{B}_j,
		\end{align}
		for $i,j \in [k]$.
		
		According to  (\ref{sharte3}), if $N \geq \min\{ 2k^2 + 2t -3, k^2 + tk + t -2\}$, then with probability approaching to one, as $|\mathbb{F}| \rightarrow \infty$, we can calculate all the coefficients of $\mathbf{H}(x)$, including $\mathbf{H}_{i-1+k(j-1)} = \mathbf{A}^T_i\mathbf{B}_j$, for $i,j \in [k]$, from $\mathbf{H}(\alpha_n)$, $n \in [N]$. In particular, there are some  $r^{(i,j)}_n$, $i,j \in [k]$ and $n\in [N]$, such that
		\begin{align}
		\label{lagrange1}
		\mathbf{A}_{i}^T\mathbf{B}_j= \sum_{n=1}^N r^{(i,j)}_n\mathbf{H}(\alpha_n). 
		\end{align}
		Note that $r^{(i,j)}_n$, $i,j \in [k]$ and $n\in [N]$ are only functions of $\alpha_n$, $ n \in [N]$, which are known by all workers. 
		Up to now  worker $n$ has access to $\mathbf{H}(\alpha_n)$. The challenge is to find a way to change the local knowledge of the $\mathbf{H}(\alpha_n)$ to  $(b,t,k)$ polynomial-share of  $\mathbf{L}$, for each worker $n$.
		\item \emph{Step 2- Communication:}
		
		Worker $n$ forms the matrix $\mathbf{H}^{(n)}$, defined as 
		\begin{align}\label{tarife q koochik1}
		\mathbf{H}^{(n)} \defeq  
		\begin{bmatrix}
		\mathbf{H}(\alpha_n) r^{(1,1)}_n & \mathbf{H}(\alpha_n) r^{(1,2)}_n & \dots & \mathbf{H}(\alpha_n) r^{(1,k)}_n \\
		\mathbf{H}(\alpha_n) r^{(2,1)}_n & \mathbf{H}(\alpha_n) r^{(2,2)}_n & \dots & \mathbf{H}(\alpha_n) r^{(2,k)}_n \\
		\vdots &‌\vdots & \vdots & \vdots \\
		\mathbf{H}(\alpha_n) r^{(k,1)}_n & \mathbf{H}(\alpha_n) r^{(k,2)}_n & \dots & \mathbf{H}(\alpha_n) r^{(k,k)}_n \\
		\end{bmatrix}.
		\end{align}
		
		Then worker $n$, $(b,t,k)$ polynomial-shares $\mathbf{H}^{(n)}$ with workers in $[N]$. More precisely, according to \eqref{tarife raveshe sharing}, worker $n$ forms the following polynomial, 
		\begin{align}
		\mathbf{F}_{\mathbf{H}^{(n)},b,t,k}(x) = \sum_{j=1}^{k}\mathbf{H}^{(n)}_jx^{b(j-1)} + \sum_{j=1}^{t-1} \mathbf{\bar{R}}^{(n)}_{j}x^{k^2+j-1},	
		\end{align}
		where $\mathbf{H}^{(n)}_j \defeq \begin{bmatrix}
		\mathbf{H}(\alpha_n) r^{(1,j)}_n\\ \mathbf{H}(\alpha_n) r^{(2,j)}_n \\ \vdots \\ \mathbf{H}(\alpha_n) r^{(k,j)}_n\end{bmatrix}$, for $j \in [k]$ and  $\mathbf{\bar{R}}^{(n)}_{j}, j = 1, 2, \dots, t-1$, are chosen independently and uniformly at random in $\mathbb{F}^{m \times \frac{m}{k}}$. Worker $n$ sends $\mathbf{F}_{\mathbf{H}^{(n)},b,t,k}(\alpha_{n'})$ to the worker $n'$, for all $n' \in [N]$. All of the workers follow the same method. Thus, at the end of this step, each worker $n'$ has access to the matrices $\{\mathbf{F}_{\mathbf{H}^{(1)},b,t,k}(\alpha_{n'}), \mathbf{F}_{\mathbf{H}^{(2)},b,t,k}(\alpha_{n'}), \dots, \mathbf{F}_{\mathbf{H}^{(N)},b,t,k}(\alpha_{n'})\}$.
		
		\item \emph{Step 3- Aggregation:}
		
		Now worker $n'$  calculates $\sum_{n=1}^{N}\mathbf{F}_{\mathbf{H}^{(n)},b,t,k}(\alpha_{n'})$.  
		We claim that this summation is indeed the polynomial share of the matrix $\mathbf{A}^T\mathbf{B}$. To verify that, consider the  polynomial function
		\begin{align}\label{tarife Q bozorg1}
		\mathbf{Q}(x) \defeq 
		\sum_{n=1}^{N}\mathbf{F}_{\mathbf{H}^{(n)},b,t,k}(x).
		\end{align}
		
		We note that 
		\begin{align}
		\mathbf{Q}(x) & = \nonumber
		\sum_{n=1}^{N}\mathbf{F}_{\mathbf{H}^{(n)},b,t,k}(x) \\ \nonumber
		& = \sum_{n=1}^{N} \sum_{j=1}^{k}\mathbf{H}^{(n)}_jx^{b(j-1)} + \sum_{n=1}^{N} \sum_{j=1}^{t-1} \mathbf{\bar{R}}^{(n)}_{j}x^{k^2+j-1}  \\ \nonumber
		& = \sum_{n=1}^{N}
		\sum_{j=1}^{k} \begin{bmatrix}
		\mathbf{H}(\alpha_n) r^{(1,j)}_n\\ \mathbf{H}(\alpha_n) r^{(2,j)}_n \\ \vdots \\ \mathbf{H}(\alpha_n) r^{(k,j)}_n\end{bmatrix}x^{b(j-1)} +
		\sum_{n=1}^{N} \sum_{j=1}^{t-1} \mathbf{\bar{R}}^{(n)}_{j}x^{k^2+j-1}	\\ \nonumber
		& = \sum_{j=1}^{k} \sum_{n=1}^{N}
		 \begin{bmatrix}
		\mathbf{H}(\alpha_n) r^{(1,j)}_n\\ \mathbf{H}(\alpha_n) r^{(2,j)}_n \\ \vdots \\ \mathbf{H}(\alpha_n) r^{(k,j)}_n\end{bmatrix}x^{b(j-1)} +
		\sum_{j=1}^{t-1} \sum_{n=1}^{N} \mathbf{\bar{R}}^{(n)}_{j}x^{k^2+j-1}	\\ \nonumber
		& \overset{(a)} {=}
		\sum_{j=1}^{k} \begin{bmatrix}
		\mathbf{A}^T_1\mathbf{B}_j\\ \mathbf{A}^T_2\mathbf{B}_j \\ \vdots \\ \mathbf{A}^T_k\mathbf{B}_j\end{bmatrix}x^{b(j-1)} +‌\sum_{j=1}^{t-1} (\sum_{n=1}^{N} \mathbf{\bar{R}}^{(n)}_{j})x^{k^2+j-1} = \\ 
		& \overset{(b)} =\sum_{j=1}^{k}\Lb_jx^{b(j-1)} + \mathbf{R}_jx^{k^2+j-1} , 
		\end{align}
		where (a) follows from \eqref{lagrange1}, and (b) follows from the definitions  $\mathbf{L}_j \defeq  \begin{bmatrix}
		\mathbf{A}^T_1\mathbf{B}_j\\ \mathbf{A}^T_2\mathbf{B}_j \\ \vdots \\ \mathbf{A}^T_k\mathbf{B}_j\end{bmatrix}$, for $j \in [k]$ and  $\mathbf{R}_j \defeq \sum_{n=1}^{N} \mathbf{\bar{R}}^{(n)}_{j} $, for $j \in [t-1]$. We note that $\Lb = \begin{bmatrix}
		\Lb_1 & \Lb_2 & \dots & \Lb_k
		\end{bmatrix}$. In addition,  $\mathbf{R}_j, j \in [t-1]$, have independent and  uniform distribution in $\mathbb{F}^{m \times \frac{m}{k}}$. Thus, $\mathbf{Q}(x)$ is in the form of $(b,t,k)$ polynomial sharing of $\mathbf{A}^T\mathbf{B}$ with workers in $[N]$, and worker $n$ has access to $\mathbf{Q}(\alpha_n)$.  These steps are detailed in Algorithm~\ref{Multiplication of Two Matrices}.
	\end{enumerate}

\begin{algorithm}[!htbp]
	\caption{Multiplication of Two Matrices}
	\label{Multiplication of Two Matrices}
	\begin{algorithmic}[1]
		\Statex
		\textbf{Inputs:}
		 Number $N$ of the workers. Matrices $\Ab$, $\Bb$ which are $(1,t,k)$ and $(k,t,k)$ polynomial-shared with workers in $[N]$, respectively. Also $\alpha_1, \alpha_2, \dots, \alpha_N \in \mathbb{F}$, $r^{(i,j)}_n$, for $i,j \in [k]$, and $n\in [N]$ are known by all workers.
		\Statex
		\textbf{Step 1- Computation:}
		\State
		Worker $n$ calculates $\mathbf{H}(\alpha_n) \buildrel \Delta \over = \mathbf{F}_{\Ab,1,t,k}^T(\alpha_n)\mathbf{F}_{\Bb,k,t,k}(\alpha_n)$.
		\Statex
		\textbf{Step 2- Communication:}
		\State
		Worker $n$, $(b,t,k)$ polynomial-shares the matrix \begin{align}
		\mathbf{H}^{(n)} \defeq  
		\begin{bmatrix}
		\mathbf{H}(\alpha_n) r^{(1,1)}_n & \mathbf{H}(\alpha_n) r^{(1,2)}_n & \dots & \mathbf{H}(\alpha_n) r^{(1,k)}_n \\
		\mathbf{H}(\alpha_n) r^{(2,1)}_n & \mathbf{H}(\alpha_n) r^{(2,2)}_n & \dots & \mathbf{H}(\alpha_n) r^{(2,k)}_n \\
		\vdots &‌\vdots & \vdots \\
		\mathbf{H}(\alpha_n) r^{(k,1)}_n & \mathbf{H}(\alpha_n) r^{(k,2)}_n & \dots & \mathbf{H}(\alpha_n) r^{(k,k)}_n \\
		\end{bmatrix} \nonumber,
		\end{align} with workers in $[N]$.
		
		\Statex
		\textbf{Step 3- Aggregation:}
		\State
		Worker $n$ calculates the sum of the messages received in the last step.

		\State
		End.
	\end{algorithmic}
\end{algorithm}

\subsection {Transposing}
Let $\mathbf{A} \in \mathbb{F}^ {m \times m}$, and
\begin{align}
\mathbf{A}=&
\begin{bmatrix} \mathbf{A}_{1}, & \mathbf{A}_{2}, & \dots, & \mathbf{A}_{k}
\end{bmatrix}, \nonumber
\end{align}
where $\mathbf{A}_i \in \mathbb{F}^{m \times \frac{m}{k}}$, for $i \in [k]$ and $k|m$. In addition, assume $\mathbf{A}$ is $(b,t,k)$ polynomial-shared with workers in $[N]$.  The goal is to $(b,t,k)$ polynomial-share $\Lb \defeq \mathbf{A}^T$with workers in $[N]$. 
We follow three steps, including splitting, communication, and aggregation.
\begin{enumerate}
	\item \emph{Step 1- Splitting:}
	
	Let us define
	\begin{align} \label{tarife fi ha}
	\mathbf{F}_i(x) \defeq \mathbf{F}_{\Ab,b,t,k}(x)(\frac{m}{k}(i-1):\frac{m}{k}i,:).
	\end{align}
	In other words, $\mathbf{F}_i(x)$,  $i \in [k]$, is a sub-matrix of  $\mathbf{F}_{\Ab,b,t,k}(x)$, including rows from $\frac{m}{k}(i-1)$ to $\frac{m}{k}i$.
	One can see that we have
	\begin{align}\label{rebeteh fiha}
	\mathbf{F}_i(x) = \sum_{j=1}^{k}\mathbf{A}_{ij}x^{b(j-1)} + \sum_{j=1}^{t-1} \mathbf{\bar{R}}_{ij}x^{k^2+j-1}, 
	\end{align}
	where 
	\begin{align*}
		\mathbf{A} \defeq
		\begin{bmatrix} 
			\mathbf{A}_{11} & \mathbf{A}_{12}& \dots &‌ \mathbf{A}_{1k} \\ 
			\mathbf{A}_{21} & \mathbf{A}_{22}& \dots &‌ \mathbf{A}_{2k} \\ 
			\vdots &‌\vdots& \dots&‌\vdots\\ 
			\mathbf{A}_{k1} & \mathbf{A}_{k2}& \dots &‌ \mathbf{A}_{kk}
		\end{bmatrix},
	\end{align*}
	and $\mathbf{A}_{ij} \in \mathbb{F}^{\frac{m}{k} \times \frac{m}{k}}$, for $i,j\in [k]$ and according to \eqref{tarife raveshe sharing}, $\mathbf{\bar{R}}_{ij}$, for $i,j\in [k]$, are independently and uniformly distributed in $\mathbb{F}^{\frac{m}{k} \times \frac{m}{k}}$.
	Since each worker $n$ has access to $\mathbf{F}_{\Ab,b,t,k}(\alpha_{n})$, it has access to $\mathbf{F}_i(\alpha_{n})$, too. Similar to (\ref{sharte1}) and \eqref{sharte2}, it can be proven that if $N \geq (k + t - 1)$, then by having $\mathbf{F}_i(\alpha_{n})$, for $n = 1, 2 ,\dots, N$, we can calculate all the coefficients of $\mathbf{F}_i(x)$, including $\mathbf{A}_{ij}$, for $j \in [k]$, with probability approaching to one, as $|\mathbb{F}| \rightarrow \infty$. More precisely, similar to (\ref{sharte1}) and \eqref{sharte2}, it can be proven that for any distinct $i_1, i_2, \dots, i_{k+t-1} \in [N]$, independently and uniformly at random chosen parameters $\alpha_1, \alpha_2, \dots, \alpha_N \in \mathbb{F}$, and $i,j \in [k]$, with probability approaching to one, as $|\mathbb{F}| \rightarrow \infty$ we have
		\begin{align*}
		H(\mathbf{A}_{ij}|\mathbf{F}_i(\alpha_{i_1}),\mathbf{F}_i(\alpha_{i_2}), \dots, \mathbf{F}_i(\alpha_{i_{k+t-1}})) = 0.
		\end{align*}
		In particular, there are some $r^{(j)}_n$, such that for $i,j \in [k]$, and $n \in [N]$,
		\begin{align}\label{lagrange fiha}
		\mathbf{A}_{ij} = \sum_{n=1}^{N} r^{(j)}_n\mathbf{F}_i(\alpha_n).
		\end{align}
		Note that $r^{(j)}_n$, $j \in [k]$ and $n\in [N]$, are only functions of $\alpha_n$, $ n \in [N]$, which are known by all workers. 
	
	\item \emph{Step 2- Communication:}
	
	Worker $n$ forms the matrix $\mathbf{H}^{(n)}$ defined as
	
	\begin{align}\label{tarife qfi koochik1}
	\mathbf{H}^{(n)} \defeq  \begin{bmatrix}
	\mathbf{F}_1(\alpha_n) r^{(1)}_n &‌\mathbf{F}_2(\alpha_n) r^{(1)}_n&‌\dots &\mathbf{F}_k(\alpha_n) r^{(1)}_n\\
	\mathbf{F}_1(\alpha_n) r^{(2)}_n &‌\mathbf{F}_2(\alpha_n) r^{(2)}_n&‌\dots &\mathbf{F}_k(\alpha_n) r^{(2)}_n\\
	\vdots&\vdots&\vdots\\
	\mathbf{F}_1(\alpha_n) r^{(k)}_n &‌\mathbf{F}_2(\alpha_n) r^{(k)}_n&‌\dots &\mathbf{F}_k(\alpha_n) r^{(k)}_n
	\end{bmatrix}.
	\end{align}
	Then worker $n$, $(b,t,k)$ polynomial-shares $\mathbf{H}^{(n)}$ with workers in $[N]$.
	More precisely, according to \eqref{tarife raveshe sharing}, worker $n$ forms the following polynomial
	\begin{align}
	\mathbf{F}_{\mathbf{H}^{(n)},b,t,k}(x) = \sum_{j=1}^{k}\mathbf{H}^{(n)}_jx^{b(j-1)} + \sum_{j=1}^{t-1} \mathbf{\hat{R}}^{(n)}_{j}x^{k^2+j-1},	
	\end{align}
	where $\mathbf{H}^{(n)}_j \defeq \begin{bmatrix}
	\mathbf{F}_j(\alpha_n) r^{(1)}_n\\ \mathbf{F}_j(\alpha_n) r^{(2)}_n \\ \vdots \\ \mathbf{F}_j(\alpha_n) r^{(k)}_n\end{bmatrix}$, for $j \in [k]$ and  $\mathbf{\hat{R}}^{(n)}_{j}, j = 1, 2, \dots, t-1$, are chosen independently and uniformly at random in $\mathbb{F}^{m \times \frac{m}{k}}$. Worker $n$ sends $\mathbf{F}_{\mathbf{H}^{(n)},b,t,k}(\alpha_{n'})$ to the worker $n'$, for all $n' \in [N]$. All of the workers follow the same method. Thus, at the end of this step, each worker $n'$ has access to the matrices $\{\mathbf{F}_{\mathbf{H}^{(1)},b,t,k}(\alpha_{n'}), \mathbf{F}_{\mathbf{H}^{(2)},b,t,k}(\alpha_{n'}), \dots, \mathbf{F}_{\mathbf{H}^{(N)},b,t,k}(\alpha_{n'})\}$.
	
	\item \emph{Step 3- Aggregation:}
	
	Now worker $n'$ calculates $\sum_{n=1}^{N}\mathbf{F}_{\mathbf{H}^{(n)},b,t,k}(\alpha_{n'})$. We claim that this summation is indeed the polynomial share of the matrix $\mathbf{A}^T$. To verify that, consider the  polynomial function
	\begin{align}
	\mathbf{Q}(x) \defeq 
	\sum_{n=1}^{N}\mathbf{F}_{\mathbf{H}^{(n)},b,t,k}(x).
	\end{align}
	
	We note that 
	\begin{align}
	\mathbf{Q}(x) & = \nonumber
	\sum_{n=1}^{N}\mathbf{F}_{\mathbf{H}^{(n)},b,t,k}(x) \\ \nonumber
	&= \sum_{n=1}^{N} \sum_{j=1}^{k}\mathbf{H}^{(n)}_jx^{b(j-1)} + \sum_{n=1}^{N} \sum_{j=1}^{t-1} \mathbf{\hat{R}}^{(n)}_{j}x^{k^2+j-1} \\ \nonumber
	& = \sum_{n=1}^{N}
	\sum_{i=1}^{k} \begin{bmatrix}
	\mathbf{F}_i(\alpha_n) r^{(1)}_n\\ \mathbf{F}_i(\alpha_n) r^{(2)}_n \\ \vdots \\ \mathbf{F}_i(\alpha_n) r^{(k)}_n\end{bmatrix}x^{b(i-1)} +
	\sum_{n=1}^{N} \sum_{j=1}^{t-1} \mathbf{\hat{R}}^{(n)}_{j}x^{k^2+j-1}	\\ \nonumber
	& =	\sum_{i=1}^{k} 
	\sum_{n=1}^{N}
	\begin{bmatrix}
	\mathbf{F}_i(\alpha_n) r^{(1)}_n\\ \mathbf{F}_i(\alpha_n) r^{(2)}_n \\ \vdots \\ \mathbf{F}_i(\alpha_n) r^{(k)}_n\end{bmatrix}x^{b(i-1)} +
	\sum_{j=1}^{t-1}
	\sum_{n=1}^{N}
	\mathbf{\hat{R}}^{(n)}_{j}x^{k^2+j-1} \\ \nonumber
	& \overset{(a)} {=}
	\sum_{i=1}^{k}  \begin{bmatrix}
	\mathbf{A}_{i1}\\ \mathbf{A}_{i2} \\ \vdots \\ \mathbf{A}_{ik}\end{bmatrix}x^{b(i-1)} +‌\sum_{j=1}^{t-1} (\sum_{n=1}^{N} \mathbf{\hat{R}}^{(n)}_{j})x^{k^2+j-1}  \\ 
	& \overset{(b)} {=} \sum_{i=1}^{k}\Lb_ix^{b(i-1)} + \sum_{j=1}^{t-1}\mathbf{R}_{j}x^{k^2+j-1} , 
	\end{align}
	where (a) follows from~\eqref{lagrange fiha}, and (b) follows from the definitions  $\mathbf{L}_i \defeq  \begin{bmatrix}
	\mathbf{A}_{i1}\\ \mathbf{A}_{i2} \\ \vdots \\ \mathbf{A}_{ik}\end{bmatrix}$, for $i \in [k]$, and  $\mathbf{R}_{j} \defeq \sum_{n=1}^{N} \mathbf{\hat{R}}^{(n)}_{j}$, for $j \in [t-1]$. We note that $\Lb = \begin{bmatrix}
	\Lb_1 & \Lb_2 & \dots & \Lb_k
	\end{bmatrix}$. In addition,  $\mathbf{R}_{j}, j \in [t-1]$, have independent and  uniform distribution in $\mathbb{F}^{m \times \frac{m}{k}}$. Thus, $\mathbf{Q}(x)$ is in the form of $(b,t,k)$ polynomial sharing of $\mathbf{A}^T$ with workers in $[N]$, and worker $n$ has access to $\mathbf{Q}(\alpha_n)$. These steps are detailed in Algorithm~\ref{Transposing}. 
\end{enumerate}

\begin{algorithm}[!htbp]
	\caption{Transposing}
	\label{Transposing}
	\begin{algorithmic}[1]
		\Statex
		\textbf{Inputs:}
		Number $N$ of the workers. Matrix $\Ab$ which is $(b,t,k)$ polynomial-shared with workers in $[N]$. Also $\alpha_1, \alpha_2, \dots, \alpha_N \in \mathbb{F}$, $r^{(j)}_n$ for $j \in [k]$, and $n \in [N]$ are known by all workers.
		\Statex
		\textbf{Step 1- Splitting:}
		\State
		Worker $n$ calculates $\mathbf{F}_i(\alpha_n) \buildrel \Delta \over = \mathbf{F}_{A,b,t,k}(\alpha_n)(\frac{m}{k}(i-1):\frac{m}{k}i,:)$, for  $i\in[k]$.
		\Statex
		\textbf{Step 2- Communication:}
		\State
		Worker $n$ $(b,t,k)$ polynomial-shares the matrix \begin{align}
		\mathbf{H}^{(n)} \defeq  \begin{bmatrix}
		\mathbf{F}_1(\alpha_n) r^{(1)}_n &‌\mathbf{F}_2(\alpha_n) r^{(1)}_n&‌\dots &\mathbf{F}_k(\alpha_n) r^{(1)}_n\\
		\mathbf{F}_1(\alpha_n) r^{(2)}_n &‌\mathbf{F}_2(\alpha_n) r^{(2)}_n&‌\dots &\mathbf{F}_k(\alpha_n) r^{(2)}_n\\
		\vdots&\vdots&\vdots\\
		\mathbf{F}_1(\alpha_n) r^{(k)}_n &‌\mathbf{F}_2(\alpha_n) r^{(k)}_n&‌\dots &\mathbf{F}_k(\alpha_n) r^{(k)}_n
		\end{bmatrix} \nonumber.
		\end{align} with workers in $[N]$.
		\Statex
		\textbf{Step 3- Aggregation:}
		\State
		Worker $n$ calculates the sum of the messages received in the last step.
		\State
		End.
	\end{algorithmic}
\end{algorithm}

\subsection{Changing the parameter of sharing}
Let $\mathbf{A} \in \mathbb{F}^ {m \times m}$, and
\begin{align}
\mathbf{A}=&
\begin{bmatrix} \mathbf{A}_{1}, & \mathbf{A}_{2}, & \dots, & \mathbf{A}_{k}
\end{bmatrix}, \nonumber
\end{align}
where $\mathbf{A}_i \in \mathbb{F}^{m \times \frac{m}{k}}$, for $i \in [k]$ and $k|m$. In addition, assume $\mathbf{A}$ is $(b,t,k)$ polynomial-shared with workers in $[N]$.
 The goal is to $(b^\prime,t,k)$ polynomial-share $\mathbf{A}$ with workers in $[N]$, where $b^\prime \neq b$.

Similar to (\ref{sharte1}) and \eqref{sharte2}, it can be proven that if $N \geq (k + t - 1)$, then by having $\mathbf{F}_{\Ab,b,t,k}(\alpha_n)$, for $n = 1, 2 ,\dots, N$, we can calculate all the coefficients of $\mathbf{F}_{\Ab,b,t,k}(x)$, including $\mathbf{A}_{j}$, for $j \in [k]$, with probability approaching to one, as $|\mathbb{F}| \rightarrow \infty$. More precisely, similar to (\ref{sharte1}) and \eqref{sharte2}, it can be proven that for any distinct $i_1, i_2, \dots, i_{k+t-1} \in [N]$, independently and uniformly at random chosen parameters $\alpha_1, \alpha_2, \dots, \alpha_N \in \mathbb{F}$, and $i,j \in [k]$, with probability approaching to one, as $|\mathbb{F}| \rightarrow \infty$ we have
	\begin{align*}
	H(\mathbf{A}_{j}|\mathbf{F}_{\Ab,b,t,k}(\alpha_{i_1}),\mathbf{F}_{\Ab,b,t,k}(\alpha_{i_{2}}), \dots, \mathbf{F}_{\Ab,b,t,k}(\alpha_{i_{k+t-1}})) = 0.
	\end{align*}
	In particular, there are some $r^{(j)}_n$, such that for $j \in [k]$, and $n \in [N]$,
	
	\begin{align}\label{lagrange fiiiiiha}
	\mathbf{A}_{j} = \sum_{n=1}^{N} r^{(j)}_n\mathbf{F}_{\Ab,b,t,k}(\alpha_n).
	\end{align}
	Note that $r^{(j)}_n$, $j \in [k]$ and $n\in [N]$, are only functions of $\alpha_n$, $ n \in [N]$, which are known by all workers. 

\begin{enumerate}
	\item \emph{Step 1- Communication:}
	
	Worker $n$ forms $\mathbf{H}^{(n)}$ defined as
	\begin{align}
	\mathbf{H}^{(n)} \defeq \begin{bmatrix}
	r^{(1)}_n \mathbf{F}_{\Ab,b,t,k}(\alpha_n), & r^{(2)}_n \mathbf{F}_{\Ab,b,t,k}(\alpha_n), & \dots, & r^{(k)}_n \mathbf{F}_{\Ab,b,t,k}(\alpha_n)
	\end{bmatrix}.	
	\end{align}
	Then worker $n$, $(b^\prime,t,k)$ polynomial-shares $\mathbf{H}^{(n)}$ with workers in $[N]$.
	More precisely, according to \eqref{tarife raveshe sharing}, worker $n$ forms the following polynomial
	\begin{align}
	\mathbf{F}_{\mathbf{H}^{(n)},b^\prime,t,k}(x) = \sum_{j=1}^{k}\mathbf{H}^{(n)}_jx^{b^\prime(j-1)} + \sum_{j=1}^{t-1} \mathbf{\bar{R}}^{(n)}_{j}x^{k^2+j-1},	
	\end{align}
	where $\mathbf{H}^{(n)}_j \defeq r^{(j)}_n \mathbf{F}_{\Ab,b,t,k}(\alpha_n)$ for $j \in [k]$, and  $\mathbf{\bar{R}}^{(n)}_{j}, j = 1, 2, \dots, t-1$, are chosen independently and uniformly at random in $\mathbb{F}^{m \times \frac{m}{k}}$. Worker $n$ sends $\mathbf{F}_{\mathbf{H}^{(n)},b^\prime,t,k}(\alpha_{n'})$ to the worker $n'$, for all $n' \in [N]$. All of the workers follow the same method. Thus, at the end of this step, each worker $n'$ has access to the matrices $\{\mathbf{F}_{\mathbf{H}^{(1)},b^\prime,t,k}(\alpha_{n'}), \mathbf{F}_{\mathbf{H}^{(2)},b^\prime,t,k}(\alpha_{n'}), \dots, \mathbf{F}_{\mathbf{H}^{(N)},b^\prime,t,k}(\alpha_{n'})\}$.
	
	\item \emph{Step 2- Aggregation:}
	
	Now worker $n'$ calculates $\sum_{n=1}^{N}\mathbf{F}_{\mathbf{H}^{(n)},b^{\prime},t,k}(\alpha_{n^{\prime}})$. We claim that this summation is indeed the polynomial share of the matrix $\mathbf{A}$. To verify that, consider the  polynomial function  
	\begin{align}\label{tarife Qfiiiii bozorg1}
	\mathbf{Q}(x) \defeq \sum_{n=1}^{N}\mathbf{F}_{\mathbf{H}^{(n)},b^{\prime},t,k}(x)
	\end{align}
	
	We note that 
	\begin{align}
	\mathbf{Q}(x) & = \nonumber
	\sum_{n=1}^{N}\mathbf{F}_{\mathbf{H}^{(n)},b^{\prime},t,k}(x) \\ \nonumber
	& = \sum_{n=1}^{N} \sum_{j=1}^{k}\mathbf{H}^{(n)}_jx^{b^{\prime}(j-1)} +  \sum_{n=1}^{N} \sum_{j=1}^{t-1} \mathbf{\bar{R}}^{(n)}_{j}x^{k^2+j-1} \\ \nonumber
	& = \sum_{n=1}^{N}
	\sum_{j=1}^{k} r^{(j)}_n \mathbf{F}_{\Ab,b,t,k}(\alpha_n)x^{b^\prime(j-1)} +
	\sum_{n=1}^{N} \sum_{j=1}^{t-1} \mathbf{\bar{R}}^{(n)}_{j}x^{k^2+j-1}	\\ \nonumber
	& = 
	\sum_{j=1}^{k} 
	\sum_{n=1}^{N}
	r^{(j)}_n \mathbf{F}_{\Ab,b,t,k}(\alpha_n)x^{b^\prime(j-1)} +
	\sum_{j=1}^{t-1}
	\sum_{n=1}^{N} 
	\mathbf{\bar{R}}^{(n)}_{j}x^{k^2+j-1}	\\ \nonumber
	& \overset{(a)} {=}
	\sum_{j=1}^{k} \mathbf{A}_{j}x^{b^\prime(j-1)} +‌\sum_{j=1}^{t-1} (\sum_{n=1}^{N} \mathbf{\bar{R}}^{(n)}_{j})x^{k^2+j-1}\\ 
	&\overset{(b)} {=}\sum_{j=1}^{k} \mathbf{A}_{j}x^{b^\prime(j-1)}  + \sum_{j=1}^{t-1}\mathbf{R}_{j}x^{k^2+j-1} , 
	\end{align}
	where (a) follows from~\eqref{lagrange fiiiiiha}, and (b) follows from the definition $\mathbf{R}_{j} \defeq \sum_{n=1}^{N} \mathbf{\bar{R}}^{(n)}_{j}$, for $j \in [t-1]$. Note that,  $\mathbf{R}_{j}, j \in [t-1]$ have independent and  uniform distribution in $\mathbb{F}^{m \times \frac{m}{k}}$. Thus, $\mathbf{Q}(x)$ 
	is in the form of $(b^\prime,t,k)$ polynomial sharing of $\Ab$ with workers in $[N]$, and worker $n$ has access to $\mathbf{Q}(\alpha_n)$.
	
	These steps are explained in Algorithm~\ref{changing the jump}. 
\end{enumerate}
		
\begin{algorithm}[!htbp]
	\caption{Changing the parameter of sharing}
	\label{changing the jump}
	\begin{algorithmic}[1]
		\Statex
		\textbf{Inputs:}
		Number $N$ of the workers. Matrix $\Ab$ which is $(b,t,k)$ polynomial-shared with workers in $[N]$. Also $\alpha_1, \alpha_2, \dots, \alpha_N \in \mathbb{F}$, $r^{(j)}_n$ for $j \in [k]$, and $n \in [N]$ are known by all workers.
		\Statex
		\textbf{Step 1- Communication:}
		\State
		Worker $n$, $(b^\prime,t,k)$ polynomial-shares the matrix 
		\begin{align}
		\mathbf{H}^{(n)}(x) \defeq \begin{bmatrix}
		r^{(1)}_n \mathbf{F}_{\Ab,b,t,k}(\alpha_n) & r^{(2)}_n \mathbf{F}_{\Ab,b,t,k}(\alpha_n) & \dots & r^{(k)}_n \mathbf{F}_{\Ab,b,t,k}(\alpha_n)
		\end{bmatrix},
		\end{align}
		 with workers in $[N]$.
		\Statex
		\textbf{Step 2- Aggregation:}
		\State
		Worker $n$ calculates the sum of the messages received in the last step.
		
		\State
		End.
	\end{algorithmic}
\end{algorithm}
	
\section{The Proposed Algorithm}
\label{LSMPC Algorithm}
As we mentioned before, in Section \ref{Procedures}, at the end of all procedures \ref{Addition}-\ref{changing the jump}, the output is in the form of \emph{polynomial sharing}. This allows us to calculate any polynomial function of the inputs by concatenating these procedures accordingly. We explain this in detail in Algorithm \ref{LSMPC}.

In the proposed algorithm, we use the arithmetic representation of the function described in Appendix \ref{Arithmetic Circuits}. In the following theorem we claim that this algorithm satisfies constraints (\ref{correctness}), (\ref{privacy for workers}), and (\ref{privacy for master}).
\begin{algorithm}[!htbp]
		\caption{Proposed Algorithm}
		\label{LSMPC}
		\begin{algorithmic}[1]
			\Statex
			Assume that $\alpha_1, \alpha_2, \dots, \alpha_N \in \mathbb{F}$ are chosen independently and uniformly at random in $\mathbb{F}$, and are available everywhere.
			\Statex
			\textbf{Phase 1- Secret Sharing:}
			Every source node $\gamma \in [\Gamma]$ takes the following steps.
			\State
			Calculates $\mathbf{F}_{\mathbf{X}^{[\gamma]},1,t,k}(x)$.
			\State
			Sends $\mathbf{F}_{\mathbf{X}^{[\gamma]},1,t,k}(\alpha_n)$ to worker $n\in[N]$.
			\Statex
			\textbf{Phase 2- Computation and communication:}
			All of the workers consider the arithmetic representation of the function according to the rules in Appendix \ref{Arithmetic Circuits}. Worker $n \in [N]$ takes the following steps.
			\State
			If it assesses all of the gates, it goes to the next phase, otherwise it considers the first non-assessed gate. Call that gate $g$.
			\label{gate}
			\State
			If $g$ is an addition gate, it follows Procedure \ref{Addition} and go to Step \ref{gate}.
			\State
			If $g$ is a multiplication by constant gate, it follows Procedure \ref{Multiplication by Constant} and go to Step \ref{gate}.
			\State
			If $g$ is a multiplication of two matrices, call the output of the gate $\Cb$, and assume 
			we have $\mathbf{C} = \mathbf{A}^T\mathbf{B}$, where $\mathbf{A}, \mathbf{B}$ are inputs of the gate.
			\State
			If it has access to $\mathbf{F}_{\Bb^T,1,t,k}(\alpha_n)$, it follows Procedure \ref{Transposing} to access $\mathbf{F}_{\Bb,1,t,k}(\alpha_n)$, and goes to Step \ref{3}.
			\State
			If it has access to $\mathbf{F}_{\Bb,1,t,k}(\alpha_n)$, it follows Procedure \ref{changing the jump} to access $\mathbf{F}_{\Bb,k,t,k}(\alpha_n)$, and goes to Step \ref{2}. \label{3}
			\State
			If it has access to $\mathbf{F}_{\Ab^T,1,t,k}(\alpha_n)$, it follows Procedure \ref{Transposing} to access $\mathbf{F}_{\Ab,1,t,k}(\alpha_n)$.\label{2}
			\State
			Follows Procedure \ref{Multiplication of Two Matrices} to access  $\mathbf{F}_{\Cb,1,t,k}(\alpha_n)$ and it goes to Step \ref{gate}.
			\Statex
			\textbf{Phase 3- Reconstruction:}
			Worker $n \in [N]$ takes the following steps.
			\State
			Stores the output of the arithmetic representation of the function in $\mathbf{O}_n$. 
			\State
			Sends $\mathbf{O}_n$ to the master.
			\State
			The master uses the method in \cite{kedlaya2011fast} to reconstruct the result using matrices $\mathbf{O}_n, n \in [N]$.
			\State
			End.
		\end{algorithmic}
\end{algorithm}

\begin{theorem}\label{qazie koli}
	Algorithm~\ref{LSMPC} satisfies constraints (\ref{correctness}), (\ref{privacy for workers}), and (\ref{privacy for master}), for any $\mathbf{X}^{[1]}, \mathbf{X}^{[2]} \dots, \mathbf{X}^{[\Gamma]}, \in \mathbb{F}^{m \times m}$, and polynomial function $\mathbf{G}: (\mathbb{F}^{m \times m})^\Gamma \rightarrow \mathbb{F}^{m \times m}$.
\end{theorem}
\begin{proof}
	To see the proof that conditions \eqref{privacy for workers} and \eqref{privacy for master} are satisfied, see Appendix \ref{Security analysis}. For constraint \eqref{correctness}, note that the master receives the polynomial sharing of the result from all of the workers, and according to Theorem \ref{qazie sharing}, these shares are enough to reconstruct the result.
\end{proof}	
\subsection{Computation and communication complexity}
In this section, we evaluate the computation and communication complexity of all the procedures and phases of Algorithm~\ref{LSMPC}, explained throughout the paper.
	\begin{itemize}
		\item Sharing phase: In this phase each source node ${\gamma} \in [{\Gamma}]$ has to evaluate $\mathbf{F}_{\mathbf{X}^{[\gamma]},1,t,k}(x)$, for $N$ distinct values. 
		According to \eqref{tarife raveshe sharing} we have
		\begin{align}
		\mathbf{F}_{\mathbf{X}^{[\gamma]},1,t,k}(x) =  \sum_{j=1}^{k}\mathbf{X}^{[\gamma]}_{j}x^{j-1} + x^{k^2}\sum_{j=1}^{t-1} \mathbf{R}_{j}x^{j-1}.
		\end{align}
		Using Horner's method \cite{wiki:Horner}, calculating a polynomial of degree $n$ can be done in $\mathcal{O}(n)$. Therefore, calculating $\sum_{j=1}^{k}\mathbf{X}^{[\gamma]}_{j}x^{j-1}$ and $x^{k^2}\sum_{j=1}^{t-1} \mathbf{R}_{j}x^{j-1}$ can be done in $\mathcal{O}(k(m \times \frac{m}{k}))$ and $\mathcal{O}((t-1)(m \times \frac{m}{k})+ 2\log{k})$, respectively. Therefore, evaluating $\mathbf{F}_{\mathbf{X}^{[\gamma]},1,t,k}(x)$ at any point can be done in $\mathcal{O}((k+t)(\frac{m^2}{k}))$. Thus, sharing phase can be done with the computation complexity of $\mathcal{O}((k+t)(\frac{m^2}{k})(N))$. 
		
		Also it has to send the resulting $m \times \frac{m}{k}$ matrices to each worker. Therefore, there is an aggregated communication complexity of $\mathcal{O}((N\Gamma)(\frac{m^2}{k}))$.
		
		\item Addition procedure: In this procedure, each worker computes only the
		addition of two matrices of dimensions $m \times \frac{m}{k}$ with the computation complexity of $\mathcal{O}(\frac{m^2}{k})$. Also, there is no communication between the workers in this procedure.
		
		\item Multiplication by constant procedure: In this procedure, each worker computes only the
		multiplication of a matrix of dimension $m \times \frac{m}{k}$ and a constant number with the computation complexity of $\mathcal{O}(\frac{m^2}{k})$. Also, there is no communication between the workers in this procedure.
		
		\item Multiplication of two matrices: In this procedure, first of all, each worker computes only the
		multiplication of two matrices of dimensions $\frac{m}{k} \times m$ and $m \times \frac{m}{k}$. If it is done conventionally, it can be executed with the computation complexity of $\mathcal{O}(\frac{m^3}{k^2})$. Then, each worker $n$ has to evaluate polynomial sharing of $\mathbf{F}_{\mathbf{H}^{(n)},b,t,k}(x)$ at $N$ points, which has a complexity of $\mathcal{O}((k+t)(\frac{m^2}{k})(N))$. Finally, each worker has to sum up $N$ matrices of dimensions $m \times \frac{m}{k}$ with the computation complexity of $\mathcal{O}((\frac{m^2}{k})(N))$. Therefore, per node computation complexity of this step is $\mathcal{O}(\frac{m^3}{k^2} + (k+t)(\frac{m^2}{k})(N) + (\frac{m^2}{k})(N)‌) = 
		\mathcal{O}(\frac{m^3}{k^2} + (k+t)(\frac{m^2}{k})(N)‌)
		$. Also there is a communication of $\mathcal{O}(N^2)$ of $m \times \frac{m}{k}$ matrices between the workers, with the aggregated communication complexity of $\mathcal{O}((N^2)(\frac{m^2}{k}))$.
		\item Transposing: In this procedure, each worker $n$ has to evaluate polynomial function of $\mathbf{F}_{\mathbf{H}^{(n)},b,t,k}(x)$ at $N$ points, which has complexity of $\mathcal{O}((k+t)(\frac{m^2}{k})(N))$. Then, each worker has to sum up $N$ matrices of dimensions $m \times \frac{m}{k}$ with the computation complexity of $\mathcal{O}((\frac{m^2}{k})(N))$. Therefore, the per node computation complexity of this step is $\mathcal{O}((k+t)(\frac{m^2}{k})(N) + (\frac{m^2}{k})(N)) = \mathcal{O}((k+t)(\frac{m^2}{k})(N))$.
		Also there is a communication of $\mathcal{O}(N^2)$ of $m \times \frac{m}{k}$ matrices between the workers, which implies that there is an aggregated communication complexity of $\mathcal{O}((N^2)(\frac{m^2}{k}))$.
		\item Changing the parameter of sharing: In this procedure, each worker $n$ has to evaluate polynomial sharing of $\mathbf{F}_{\mathbf{H}^{(n)},b^{\prime},t,k}(x)$ at $N$ points, which has complexity of $\mathcal{O}((k+t)(\frac{m^2}{k})(N))$. Then, each worker has to sum up $N$ matrices of dimensions $m \times \frac{m}{k}$ with the computation complexity of $\mathcal{O}((\frac{m^2}{k})(N))$. Therefore, the per node computation complexity of this step is $\mathcal{O}((k+t)(\frac{m^2}{k})(N) + (\frac{m^2}{k})(N)) = \mathcal{O}((k+t)(\frac{m^2}{k})(N))$.
		Also there is a communication of $\mathcal{O}(N^2)$ of $m \times \frac{m}{k}$ matrices between the workers, with the aggregated communication complexity of $\mathcal{O}((N^2)(\frac{m^2}{k}))$.
		\item Reconstruction phase: In this phase, the master has to reconstruct the value of the result, using the messages it has received from the workers. Since the result is in the form of polynomial sharing at the master, we have to interpolate a polynomial of degree $k^2+t-2$, which can be done with the complexity of $\mathcal{O}((k^2+t-2)\log^2{(k^2+t-2)}\log{\log{(k^2+t-2)}})$, according to \cite{kedlaya2011fast}. Also, there is a communication of $N$ of $m \times \frac{m}{k}$ matrices between the workers and master, with the aggregated communication complexity of $N\frac{m^2}{k}$.
	\end{itemize}
	
	According to the above calculations, multiplication of two matrices has the most computation complexity among the all of the procedures described throughout the paper. Now let us calculate the computation complexity of the function $G$. Assume that $v$ is the number of monomial terms of $\mathbf{G}$, and $d = \deg \mathbf{G}$. Therefore, the per node computation complexity of the  function $\mathbf{G}$ is at most $\mathcal{O}((vd)(\frac{m^3}{k^2} + (k+t)(\frac{m^2}{k})N)‌)$. Also the aggregated communication complexity of the  function $\mathbf{G}$ is at most  $\mathcal{O}(N\Gamma(\frac{m^2}{k})+dN^2(\frac{m^2}{k})+N(\frac{m^2}{k})) = \mathcal{O}(N\frac{m^2}{k}(dN + \Gamma))$. 
	
	In the context of MPC, there are some solutions that decrease the communication complexity, which are not information theoretic \cite{damgaard2013practical, couteau2019note, yoshida2018efficiency, choudhury2016efficient,beaver1991efficient}. One idea that is worth exploring is to combine the idea of this work and those solutions.

\section{Extension} \label{Extension}
In the proposed scheme in Section \ref{LSMPC Algorithm}, in order to share the matrix according to the \emph{polynomial sharing} scheme, we partition it column-wise. This model of partitioning and sharing can be extended. In general, we can partition the matrix into some blocks and share the matrix according to this configuration. This approach is inspired and motivated by \emph{entangled polynomial code}~\cite{entangle}, or \emph{MatDot code}~\cite{opt-recovery}.
\begin{definition}	
	Let $\mathbf{A} \in \mathbb{F}^{z \times v}$, for some $z,v \in \mathbb{N}$, be
	\begin{flalign}
	\nonumber
	&\mathbf{A}=\begin{bmatrix}
	\mathbf{A}_{0, 0} & \mathbf{A}_{0, 1} & \mathbf{A}_{0, 2} &\dots & \mathbf{A}_{0, m-1}\\
	\mathbf{A}_{1, 0} & \mathbf{A}_{1, 1} & \mathbf{A}_{1, 2} &\dots & \mathbf{A}_{1 ,m-1}\\
	\mathbf{A}_{2, 0} & \mathbf{A}_{2, 1} & \mathbf{A}_{2, 2} &\dots & \mathbf{A}_{2, m-1}\\
	\vdots & \vdots & \vdots & \ddots & \vdots \\
	\mathbf{A}_{p-1, 0} & \mathbf{A}_{p-1, 1} & \mathbf{A}_{p-1, 2} &\dots & \mathbf{A}_{p-1, m-1}
	\end{bmatrix},
	\end{flalign}
	where $\mathbf{A}_{i,j} \in \mathbb{F}^{k \times k'}$, for some $k,k' \in \mathbb{N}$, $k|z$, and $k'|v$.
	
	We define the entangled polynomial function $\mathbf{F}_{\mathbf{A},b,p,m,n,t,\textrm{s}} (x)$,  for some $p,m,n,t \in \mathbb{N}$, as
	\small
	\begin{align*}
	\mathbf{F}_{\mathbf{A},b,p,m,n,t,+}(x)& \defeq
	\displaystyle\sum_{j=0}^{p-1} \displaystyle\sum_{k=0}^{m-1} \mathbf{A}_{j,k}x^{j+bkp} +\displaystyle\sum_{q=0}^{t-2}\mathbf{R}'_qx^{npm+q} ,\\
	\mathbf{F}_{\mathbf{A},b,p,m,n,t,-}(x)& \defeq
	\displaystyle\sum_{j=0}^{p-1} \displaystyle\sum_{k=0}^{m-1} \mathbf{A}_{j,k}x^{(p-1-j)+bkp} +\sum_{q=0}^{t-2}\mathbf{R}'_qx^{npm+q},
	\end{align*}
	\normalsize
	where $\mathbf{R}'_q$, $q \in \{0,1,...,t-2\}$, are chosen independently and uniformly at random from $\mathbb{F}^{k \times k'}$.
\end{definition}
Now with this method of sharing we can do more general models of polynomial calculation (see~\cite{Multi-Party2}).

\section{Conclusion and Discussion}
In this paper, we developed a new secure multiparty computation for massive input data. The proposed solution offers significant gains compared to schemes based on splitting the data into smaller pieces and applying conventional multiparty computation. In this work, we assumed that some of the nodes are semi-honest. The next step is to consider the case where nodes are adversarial, which has been addressed in~\cite{seyed2019}. There are many open problems in this direction. This includes exploring communication efficiency, the tradeoff between the number of servers and communication load, having a network of heterogeneous servers,  various network topologies, and the cases where some communication links are eavesdropped.  
In addition, investigating the case where the sources are collocated and can be encoded together would be interesting. Here we assume that we want to calculate  $\mathbf{G}(\mathbf{X}^{[1]}, \mathbf{X}^{[2]} \dots, \mathbf{X}^{[\Gamma]})$, where inputs are massive. One interesting direction is to consider the case, where the goal is to calculate $\mathbf{G}(\mathbf{X}^{[1]}_i, \mathbf{X}^{[2]}_i \dots, \mathbf{X}^{[\Gamma]}_i)$, for $i=1, \ldots, k$, for some integer $k$. This would be in the intersection of this work and~\cite{lagrange}. Exploiting the sparsity of the input data in this calculation would also be of great interest (see~\cite{codedsketch}).

\bibliographystyle{ieeetr}
\bibliography{journal_abbr,SDMPC}

\begin{appendices}
	\clearpage
	\section{Proof of Theorem \ref{qazie sharing}} \label{Appendice A}
	In order to prove Theorem \ref{qazie sharing}, we first prove the following lemma.
		\begin{lemma}
			\label{lemm tedade nasefr}
			Let $\mathbf{A}, \mathbf{B} \in \mathbb{F}^{m \times m}$, and
			\begin{align}
			\mathbf{A}&=
			\begin{bmatrix} \mathbf{A}_{1}, & \mathbf{A}_{2}, & \dots, & \mathbf{A}_{k}
			\end{bmatrix}, \nonumber\\ 
			\mathbf{B}&= 
			\begin{bmatrix} \mathbf{B}_{1}, & \mathbf{B}_{2}, & \dots, & \mathbf{B}_{k}
			\end{bmatrix},	\nonumber 
			\end{align}
			where $\mathbf{A}_i, \mathbf{B}_i \in \mathbb{F}^{m \times \frac{m}{k}}$, for $i \in [k]$ and $k|m$.
			In addition, assume that these matrices are shared among  $N$ workers using polynomial functions $\mathbf{F}_{\Ab,1,t,k}(x)$ and $\mathbf{F}_{\Bb,k,t,k}(x)$, respectively. Let us define
			\begin{align}
			\mathbf{H}(x) = \sum_{n=0}^{2(k^2+t-2)} \mathbf{H}_nx^n \defeq \mathbf{F}_{\Ab,1,t,k}^T(x)\mathbf{F}_{\Bb,k,t,k}(x).
			\end{align}
			The number of nonzero coefficients of $\mathbf{H}(x)$ is equal to
			\[   
			\min \{2k^2 + 2t - 3, k^2 + kt + t -2\} = 
			\begin{cases}
			2k^2 + 2t - 3 	&\quad\text{if} ~k < t \\
			k^2 + kt + t -2 &\quad\text{if} ~k \geq t
			\end{cases}.
			\]
		\end{lemma}
		\begin{proof}
			We have
			\begin{align}
			\mathbf{F}_{\Ab,1,t,k}^T(x) &= \sum_{n=1}^{k}\mathbf{A}^T_nx^{n-1} + x^{k^2}\sum_{n=1}^{t-1} \mathbf{\bar{R}}^T_{n}x^{n-1}, \nonumber \\
			\mathbf{F}_{\Bb,k,t,k}(x) &= \sum_{n=1}^{k}\mathbf{B}_nx^{k(n-1)} + x^{k^2}\sum_{n=1}^{t-1} \mathbf{\hat{R}}_{n}x^{n-1}. \nonumber
			\end{align}
			\begin{itemize}
				\item 	The power of $x$ with nonzero coefficients in $(\sum_{n=1}^{k}\mathbf{A}^T_nx^{n-1})(\sum_{n=1}^{k}\mathbf{B}_nx^{k(n-1)})$, are 
				\begin{align*}
				\mathcal{S}_1 &\defeq \{0,1,2,\dots, k^2-1\}.
				\end{align*}
				
				\item The power of $x$ with nonzero coefficients in $(x^{k^2}\sum_{n=1}^{k}\mathbf{A}^T_nx^{n-1})(\sum_{n=1}^{t-1} \mathbf{\hat{R}}_{n}x^{n-1})$, are
				\begin{align*}
				\mathcal{S}_2 &\defeq \{k^2,k^2 + 1,k^2 + 2,\dots, k^2 + k - 1 + t -2\}.
				\end{align*}
				
				\item The power of $x$ with nonzero coefficients in $(x^{k^2}\sum_{n=1}^{t-1} \mathbf{\bar{R}}^T_{n}x^{n-1})(\sum_{n=1}^{k}\mathbf{B}_nx^{k(n-1)})$, are
				\begin{align*}
				\mathcal{S}_3 &\defeq \{k^2 + ik + j, i \in [0,k-1], j \in [0,t-2]\}.
				\end{align*}
				
				\item The power of $x$ with nonzero coefficients in $(x^{2k^2}\sum_{n=1}^{t-1} \mathbf{\bar{R}}^T_{n}x^{n-1})(\sum_{n=1}^{t-1} \mathbf{\hat{R}}_{n}x^{n-1})$, are 
				\begin{align*}
				\mathcal{S}_4 &\defeq \{2k^2,2k^2 + 1,2k^2 + 2,\dots, 2k^2 + 2t - 4\}.
				\end{align*}
			\end{itemize}
			The degree of the polynomial $\mathbf{H}(x)$ is $2k^2 + 2t - 4$. Therefore, there are $2k^2 + 2t - 3$ coefficients from $0$ to $2k^2 + 2t - 4$, where some of them are zero. One can see that the set of zero coefficients of $\mathbf{H}(x)$ is equal to
			\begin{align}
			[0, 2k^2 + 2t - 4] - (\mathcal{S}_1 \cup \mathcal{S}_{2} \cup \mathcal{S}_3 \cup  \mathcal{S}_{4}). \label{aslekari}
			\end{align}
			
			We consider the following two cases.
			\begin{enumerate}
				\item Case 1, $k-1 \leq t-2$: In this case one can see that we have
				\begin{align*}
				(\mathcal{S}_1 \cup \mathcal{S}_{2} \cup \mathcal{S}_3 \cup  \mathcal{S}_{4}) = [0, 2k^2 + 2t - 4].
				\end{align*}
				Therefore, in this case non of the coefficients of $\mathbf{H}(x)$, is equal to zero. Thus, the number of nonzero coefficients of $\mathbf{H}(x)$ is $2k^2 + 2t - 3$.
				
				\item Case 2, $k-1 > t-2$: In this case, counting the number of non-zero coefficients is more complicated, specially because the intersection $\mathcal{S}_2 \cap \mathcal{S}_3$ is not zero. In this case we claim that the number of zero coefficients of $\mathbf{H}(x)$ is $(k-t+1)(k-1)$, thus the number of nonzero coefficients is $(2k^2 +2t-3) - (k-t+1)(k-1) = k^2 + kt + t -2$. 
				
				Let us define
				\begin{align*}
				\mathcal{S}_{21} &= \{ k^2,k^2 + 1,k^2 + 2,\dots, k^2 + k - 1 \}, \\
				\mathcal{S}_{22} &= \mathcal{S}_{2} - \mathcal{S}_{21}, \\
				\mathcal{S}_{3i} &= \{k^2 + ik + j, j \in [0,t-2]\},
				\end{align*}
				for $i \in [0, k-1]$.
				
				In this case $(k-1 > t-2)$, on can see that we have
				\begin{align}
				\mathcal{S}_{2} &= \mathcal{S}_{21} \cup \mathcal{S}_{22}, \label{s22} \\
				\mathcal{S}_3 &= \cup_{i=0}^{k-1} \mathcal{S}_{3i}, \nonumber \\
				\mathcal{S}_{30} &\subset \mathcal{S}_{21}, \label{s3zires2} \\
				\mathcal{S}_{3} \cap \mathcal{S}_{21}&= \mathcal{S}_{30}, \nonumber \\
				\mathcal{S}_{22} &\subset \mathcal{S}_{31}, \nonumber \\
				\mathcal{S}_1 \cup \mathcal{S}_{21} &= [0, k^2+k-1], \nonumber \\
				\mathcal{S}_1 \cup \mathcal{S}_{21} \cup \mathcal{S}_4 &= [0, k^2+k-1] \cup [2k^2, 2k^2+2t-4], \nonumber \\
				(\mathcal{S}_1 \cup \mathcal{S}_{21} \cup \mathcal{S}_4) \cap (\mathcal{S}_{22} \cup  (\mathcal{S}_{3} - \mathcal{S}_{30})) & = \emptyset, \label{mohem} \\
				\mathcal{S}_{3i} \cap \mathcal{S}_{3j} &= \emptyset, \nonumber
				\end{align}
				for distinct $i,j \in [0, k-1]$.
				
				It is important to note that 
				\begin{align}
				(\mathcal{S}_1 \cup \mathcal{S}_{2} \cup \mathcal{S}_3 \cup  \mathcal{S}_{4}) &\overset{(a)}= (\mathcal{S}_1 \cup (\mathcal{S}_{21} \cup \mathcal{S}_{22}) \cup \mathcal{S}_3 \cup  \mathcal{S}_{4}) \nonumber\\
				&= ((\mathcal{S}_1 \cup \mathcal{S}_{21} \cup  \mathcal{S}_{4}) \cup (\mathcal{S}_3 \cup \mathcal{S}_{22})) \nonumber\\
				&= ((\mathcal{S}_1 \cup \mathcal{S}_{21} \cup  \mathcal{S}_{4}) \cup ((\mathcal{S}_3 - \mathcal{S}_{30}) \cup \mathcal{S}_{30} \cup \mathcal{S}_{22})) \nonumber\\
				&= ((\mathcal{S}_1 \cup (\mathcal{S}_{21} \cup \mathcal{S}_{30}) \cup  \mathcal{S}_{4}) \cup ((\mathcal{S}_3 - \mathcal{S}_{30}) \cup \mathcal{S}_{22})) \nonumber\\
				&\overset{(b)}= ((\mathcal{S}_1 \cup \mathcal{S}_{21} \cup \mathcal{S}_{4}) \cup (\mathcal{S}_3 - \mathcal{S}_{30})), \label{khalebadjoor}
				\end{align} 
				where (a) follows from \eqref{s22} and (b) follows from \eqref{s3zires2} and the fact that $\mathcal{S}_{22} \subset \mathcal{S}_{31} \subset (\mathcal{S}_3 - \mathcal{S}_{30})$. Also note that
				\begin{align}\label{non-zero}
				[0, 2k^2 + 2t - 4] - (\mathcal{S}_1 \cup \mathcal{S}_{21} \cup \mathcal{S}_4)  = \{k^2+k, k^2+k+1, \dots, 2k^2-1\}. 
				\end{align}
				Thus, according to \eqref{aslekari}, \eqref{mohem}, and \eqref{khalebadjoor}, in order to calculate the number of zero-coefficients, we must exclude $(\mathcal{S}_3 - \mathcal{S}_{30})$ from the set $\{k^2+k, k^2+k+1, \dots, 2k^2-1\}$. Therefore, the number of zero-coefficients of $\mathbf{H}(x)$ is equal to
				\begin{align*}
				|\{k^2+k, k^2+k+1, \dots, 2k^2-1\} - (\mathcal{S}_3 - \mathcal{S}_{30})| 
				&\overset{(a)} = (k^2 - k) - (t-1)(k-1) \\
				&= k(k-1) - (t-1)(k-1) \\
				&= (k-t+1)(k-1)
				\end{align*}
				where (a) follows from $|\{k^2+k, k^2+k+1, \dots, 2k^2-1\}| = k^2 - k$, $|\mathcal{S}_{3} - \mathcal{S}_{30}| = (t-1)(k-1)$, and $(\mathcal{S}_{3} - \mathcal{S}_{30} \subset \{k^2+k, k^2+k+1, \dots, 2k^2-1\})$.
			\end{enumerate}
		\end{proof}
		Now we prove Theorem \ref{qazie sharing}. First note that based on Lemma \ref{lemm tedade nasefr}, the number of nonzero coefficients of $\mathbf{H}(x)$ is 
		\[   
		\min \{2k^2 + 2t - 3, k^2 + kt + t -2\} = 
		\begin{cases}
		2k^2 + 2t - 3 	&\quad\text{if} ~k < t \\
		k^2 + kt + t -2 &\quad\text{if} ~k \geq t
		\end{cases}.
		\]
		Assume that $N = \min \{2k^2 + 2t - 3, k^2 + kt + t -2\}$, and consider distinct values $\alpha_1, \alpha_2, \dots, \alpha_N \in \mathbb{F}$. Let us define
		\begin{align*}
		\mathbf{H} \defeq
		\begin{bmatrix}
		\alpha_1^{j_1} & \alpha_1^{j_2} & \dots & \alpha_1^{j_N} \\
		\alpha_2^{j_1} & \alpha_1^{j_2} & \dots & \alpha_2^{j_N} \\
		\vdots & \vdots & \vdots & \vdots \\
		\alpha_N^{j_1} & \alpha_N^{j_2} & \dots & \alpha_N^{j_N}
		\end{bmatrix},
		\end{align*}
		where $j_1, j_2, \dots, j_N \in \{0, 1, \dots, 2(k^2+t-2)\}$ are the distinct indexes of the nonzero coefficients of $\mathbf{H}(x)$. Also, for any distinct $i_1 < i_2 < \dots < i_{k+t-1} \in [N]$, define
		\begin{align*}
		\mathbf{A}_{i_1, i_2, \dots, i_{k+t-1}}& \defeq
		\begin{mpmatrix}
		\alpha_{i_1}^{0} & \alpha_{i_1}^{1} & \dots & \alpha_{i_1}^{k-1} &  \alpha_{i_1}^{k^2} & \alpha_{i_1}^{k^2+1} & \dots & \alpha_{i_1}^{k^2 + t -2}\\
		\alpha_{i_2}^{0} & \alpha_{i_2}^{1} & \dots & \alpha_{i_2}^{k-1} &  \alpha_{i_2}^{k^2} & \alpha_{i_2}^{k^2+1} & \dots & \alpha_{i_2}^{k^2 + t -2} \\
		\vdots & \vdots & \vdots & \vdots & \vdots & \vdots & \vdots & \vdots \\
		\alpha_{i_{k+t-1}}^{0} & \alpha_{i_{k+t-1}}^{1} & \dots & \alpha_{i_{k+t-1}}^{k-1} &  \alpha_{i_{k+t-1}}^{k^2} & \alpha_{i_{k+t-1}}^{k^2+1} & \dots & \alpha_{i_{k+t-1}}^{k^2 + t -2}
		\end{mpmatrix}, \\
		\mathbf{B}_{i_1, i_2, \dots, i_{k+t-1}}& \defeq
		\begin{mpmatrix}
		\alpha_{i_1}^{0} & \alpha_{i_1}^{k} & \dots & \alpha_{i_1}^{(k-1)k} &  \alpha_{i_1}^{k^2} & \alpha_{i_1}^{k^2+1} & \dots & \alpha_{i_1}^{k^2 + t -2}\\
		\alpha_{i_2}^{0} & \alpha_{i_2}^{k} & \dots & \alpha_{i_2}^{(k-1)k} &  \alpha_{i_2}^{k^2} & \alpha_{i_2}^{k^2+1} & \dots & \alpha_{i_2}^{k^2 + t -2} \\
		\vdots & \vdots & \vdots & \vdots & \vdots & \vdots & \vdots & \vdots \\
		\alpha_{i_{k+t-1}}^{0} & \alpha_{i_{k+t-1}}^{k} & \dots & \alpha_{i_{k+t-1}}^{(k-1)k} &  \alpha_{i_{k+t-1}}^{k^2} & \alpha_{i_{k+t-1}}^{k^2+1} & \dots & \alpha_{i_{k+t-1}}^{k^2 + t -2}
		\end{mpmatrix}.
		\end{align*}
		These three matrices are  called Generalized Vandermonde matrix \cite{sobczyk2002generalized, kitamoto2014computation}, which has been extensively studied in the literatures.  In \cite{kitamoto2014computation} it has been shown that the determinant of the generalized Vandermonde matrix $\mathbf{A}_{i_1, i_2, \dots, i_{k+t-1}}$ is 
		\begin{align*}
		\det \mathbf{A}_{i_1, i_2, \dots, i_{k+t-1}} = g(\alpha_{i_1}, \alpha_{i_2}, \dots, \alpha_{i_{k+t-1}})\prod_{i_j > i_l}(\alpha_{i_j} - \alpha_{i_l}),
		\end{align*}
		where  the function $g(\alpha_{i_1}, \alpha_{i_2}, \dots, \alpha_{i_{k+t-1}})$ is called a Schur polynomial\cite{wiki:schur}.  Unfortunately, there is no guarantee that $g(\alpha_{i_1}, \alpha_{i_2}, \dots, \alpha_{i_{k+t-1}})$ is not equal to zero, given that $\alpha_j$'s are distinct. Thus we cannot say that $\mathbf{A}_{i_1, i_2, \dots, i_{k+t-1}}$ is full rank. 
		
		To prove Theorem \ref{qazie sharing}, its enough to show that for  large enough $|\mathbb{F}|$, there exist distinct values $\alpha_1, \alpha_2, \dots, \alpha_N \in \mathbb{F}$, such that for any distinct $i_1 < i_2 < \dots < i_{k+t-1} \in [N]$, we have
		\begin{align}
		\det{\mathbf{A}_{i_1, i_2, \dots, i_{k+t-1}}}& \neq 0, \label{11}\\
		\det{\mathbf{B}_{i_1, i_2, \dots, i_{k+t-1}}}& \neq 0, \label{22}\\
		\det{\mathbf{H}}& \neq 0, \label{33}
		\end{align}
		and if we choose $\alpha_1, \alpha_2, \dots, \alpha_N$, independently and uniformly at random in $\mathbb{F}$, the probability that (\ref{11}), (\ref{22}), and (\ref{33}) hold, approaches to one, as $|\mathbb{F}| \rightarrow \infty$. 
		
		Let us define 
		\begin{align*}
		f(\alpha_1, \alpha_2, \dots, \alpha_N) \defeq 
		(\prod_{i_1, i_2, \dots, i_{k+t-1} \in [N]} \det \mathbf{A}_{i_1, i_2, \dots, i_{k+t-1}}\det \mathbf{B}_{i_1, i_2, \dots, i_{k+t-1}}) \det \mathbf{H}.		
		\end{align*}
		This polynomial is not equal to zero polynomial. Thus, according to \cite{koetter2003algebraic}, for  large enough $|\mathbb{F}|$, there exist distinct $\alpha_1, \alpha_2, \dots, \alpha_N \in \mathbb{F}$, such that $f(\alpha_1, \alpha_2, \dots, \alpha_N) \neq 0$.
		Also, based on Schwartz-Zippel Lemma \cite{schwartz1980fast, zippel1979probabilistic}, if we choose $\alpha_1, \alpha_2, \dots, \alpha_N$, independently and uniformly at random in $\mathbb{F}$, the probability that (\ref{11}), (\ref{22}), and (\ref{33}) hold, approaches to one, as $|\mathbb{F}| \rightarrow \infty$.

	\clearpage
	\section{Proof of Privacy In Theorem \ref{qazie koli}} \label{Security analysis}
	Recall that in this algorithm, to share any information, we always add some random matrices to it. We claim that this protocol satisfies privacy constraints \eqref{privacy for workers} and \eqref{privacy for master}. In order to formally prove that, we use the following two  lemmas.
	\begin{lemma}
		\label{independence lemma}
		Consider the polynomial $r(x)‌ \from \mathbb{F} \to \mathbb{F}$ 
		\begin{align*}
		r(x) = \sum_{n=1}^{t-1} a_nx^{n-1}, \nonumber	
		\end{align*}
		where $a_1, a_2, \dots, a_{t-1}$, are chosen independently and uniformly at random in $\mathbb{F}$.
		 Define
		\begin{align}
		\mathbf{\tilde{r}} \defeq (r(\alpha_1), r(\alpha_2), \dots, r(\alpha_{t-1})), \nonumber
		\end{align}
		 for some distinct values $\alpha_1, \alpha_2. \dots, \alpha_{t-1} \in \mathbb{F}$. Then $\mathbf{\tilde{r}}$ has a uniform distribution over $\mathbb{F}^{t-1}$.
	\end{lemma}
	\begin{proof}
		Assume that $r_1, r_2, \dots, r_{t-1}$ are chosen independently and uniformly at random in $\mathbb{F}$. According to Lagrange interpolation rule \cite{bakhvalov1977numerical}, we know that the following set of equations 
		\begin{align}
		\begin{array}{rl}
		r(\alpha_1) = & r_1,\\
		r(\alpha_2) = & r_2,\\
		\vdots &  \\
		r(\alpha_{t-1}) = & r_{t-1},
		\end{array} \nonumber 
		\end{align}
		 has a unique answer. 
		It means that if we know the values of $r_1, r_2, \dots, r_{t-1}$, we can uniquely determine the values of $a_i$, for $i \in [t-1]$. Also it is obvious that if we know the values of $a_1, a_2, \dots, a_{t-1}$, we can uniquely determine the values of $r(\alpha_1), r(\alpha_2), \dots, r(\alpha_{t-1})$. Therefore, there is a one to one mapping between  $\mathbf{a} = (a_1, a_2, \dots, a_{t-1})$ and $\mathbf{\tilde{r}} = (r(\alpha_1), r(\alpha_2), \dots, r(\alpha_{t-1}))$. Note that $a_i, i=1, 2, \dots, t-1$, are chosen independently and uniformly at random in $\mathbb{F}$, thus $\mathbf{a}$ has uniform distribution over $\mathbb{F}^{t-1}$. Therefore, $\mathbf{\tilde{r}}$ has uniform distribution over $\mathbb{F}^{t-1}$, too.
	\end{proof}
	\begin{corollary}
		\label{natije asli randomness}
		Assume that $\mathbf{R}(x) \from \mathbb{F} \to \mathbb{F}^{p \times q}$ is a polynomial function of degree $\max(0,t-2), t \in \mathbb{N}$, where the $t-1$ coefficients are chosen uniformly at random in $\mathbb{F}^ {p \times q}$. Define $\mathbf{\tilde{R}}$ as
		\begin{align}
		\mathbf{\tilde{R}} \defeq (\mathbf{R}(\alpha_1), \mathbf{R}(\alpha_2), \dots,\mathbf{R}(\alpha_{t-1})), \nonumber
		\end{align}
		where the $\alpha_1, \alpha_2, \dots, \alpha_{t-1} \in \mathbb{F}$ are distinct values. Then $\mathbf{\tilde{R}}$ has a uniform distribution over $\mathbb{F}^{p \times (t-1)q}$.
	\end{corollary}
	\begin{proof}
		The proof directly follows from Lemma \ref{independence lemma}.
	\end{proof}

	\begin{lemma}
			\label{lemme security}
			Consider $m$ polynomials $\mathbf{U}_1(x), \mathbf{U}_2(x), \dots, \mathbf{U}_m(x)$ of degree at most $n-1$ where their coefficients are chosen with arbitrary joint distribution from $\mathbb{F}^{p \times q}$. Let $\mathcal{A}$ denotes the order set of those coefficients. Consider the polynomials
			\begin{align}
			\mathbf{T}_1(x) &= \mathbf{U}_1(x) + x^{n}\mathbf{R}_1(x),   \nonumber \\
			\mathbf{T}_2(x) &= \mathbf{U}_2(x) + x^{n}\mathbf{R}_2(x), \nonumber  \\
			\vdots&  \nonumber\\
			\mathbf{T}_m(x) &= \mathbf{U}_m(x) + x^{n}\mathbf{R}_m(x),	
			\end{align}
			where for $1 \leq i \leq m$, $\mathbf{R}_i(x) \from \mathbb{F} \to \mathbb{F}^{p \times q}$ is a polynomial of degree $\max(0,t-2)$, where the $t-1$ coefficients are chosen independently and uniformly at random from $\mathbb{F}^{p \times q}$. Then, $I(\mathcal{A};\mathbf{\tilde{T}}) = 0$, where $\mathbf{\tilde{T}}$ defined as
			\begin{align} 
			\mathbf{\tilde{T}}
			&\defeq  \begin{bmatrix}
			\mathbf{T}_{1}(\beta_{1}) & \mathbf{T}_{1}(\beta_{2}) & \dots & \mathbf{T}_{1}(\beta_{t-1}) \\
			\mathbf{T}_{2}(\beta_{1}) & \mathbf{T}_{2}(\beta_{2}) & \dots & \mathbf{T}_{2}(\beta_{t-1}) \\
			\vdots & \vdots & \vdots & \vdots \\
			\mathbf{T}_{m}(\beta_{1}) & \mathbf{T}_{m}(\beta_{2}) & \dots & \mathbf{T}_{m}(\beta_{t-1})
			\end{bmatrix},  \\	
			\end{align}
			for some arbitrary values $\beta_{1}, \beta_{2}, \dots, \beta_{t-1} \in \mathbb{F}$.
		\end{lemma}
		\begin{proof}
			Let us define 
			\begin{align}
			\mathbf{\tilde{U}}
			& \defeq \begin{bmatrix}
			\mathbf{U}_{1}(\beta_{1}) & \mathbf{U}_{1}(\beta_{2}) & \dots & \mathbf{U}_{1}(\beta_{t-1}) \\
			\mathbf{U}_{2}(\beta_{1}) & \mathbf{U}_{2}(\beta_{2}) & \dots & \mathbf{U}_{2}(\beta_{t-1}) \\
			\vdots & \vdots & \vdots & \vdots \\
			\mathbf{U}_{m}(\beta_{1}) & \mathbf{U}_{m}(\beta_{2}) & \dots & \mathbf{U}_{m}(\beta_{t-1})
			\end{bmatrix},  \nonumber \\
			\mathbf{\tilde{R}}
			&\defeq 
			\begin{bmatrix}
			\mathbf{R}_{1}(\beta_{1}) & \mathbf{R}_{1}(\beta_{2}) & \dots & \mathbf{R}_{1}(\beta_{t-1}) \\
			\mathbf{R}_{2}(\beta_{1}) & \mathbf{R}_{2}(\beta_{2}) & \dots & \mathbf{R}_{2}(\beta_{t-1}) \\
			\vdots & \vdots & \vdots & \vdots \\
			\mathbf{R}_{m}(\beta_{1}) & \mathbf{R}_{m}(\beta_{2}) & \dots & \mathbf{R}_{m}(\beta_{t-1}) 
			\end{bmatrix}. 	\nonumber
			\end{align}
			
			For any $\mathbf{T}, \mathbf{U}, \in \mathbb{F}^{mp \times (t-1)q}$ we have
			\begin{align}
			\Pr(\mathbf{\tilde{U}} = \mathbf{U} | \mathbf{\tilde{T}} = \mathbf{T}) 
			&\overset{(a)}= \frac{\Pr(\mathbf{\tilde{T}} = \mathbf{T} | \mathbf{\tilde{U}} = \mathbf{U})\Pr(\mathbf{\tilde{U}} = \mathbf{U})}{\sum_{\mathbf{U}_j \in \mathbb{F}^{mp \times (t-1)q}} \Pr(\mathbf{\tilde{T}} = \mathbf{T} | \mathbf{\tilde{U}} = \mathbf{U}_j) \Pr(\mathbf{\tilde{U}} = \mathbf{U}_j)}    \nonumber \\
			&= \frac{\Pr(\mathbf{\tilde{R}} = \mathbf{T} -  \mathbf{U} | \mathbf{\tilde{U}} = \mathbf{U})\Pr(\mathbf{\tilde{U}} = \mathbf{U})}{\sum_{\mathbf{U}_j \in \mathbb{F}^{mp \times (t-1)q}} \Pr(\mathbf{\tilde{R}} = \mathbf{T}-\mathbf{U}_j | \mathbf{\tilde{U}} = \mathbf{U}_j)\Pr(\mathbf{\tilde{U}} = \mathbf{U}_j)}  \nonumber\\
			& \overset{(b)}= \frac{\Pr(\mathbf{\tilde{R}} = \mathbf{T} -  \mathbf{U}) \Pr(\mathbf{\tilde{U}} = \mathbf{U})}{\sum_{\mathbf{U}_j \in \mathbb{F}^{mp \times (t-1)q}} \Pr(\mathbf{\tilde{R}} = \mathbf{T} -  \mathbf{U}) \Pr(\mathbf{\tilde{U}} = \mathbf{U}_j)}   \nonumber \\
			& = \frac{\Pr(\mathbf{\tilde{U}} = \mathbf{U})}{\sum_{\mathbf{U}_j \in \mathbb{F}^{mp \times (t-1)q}} \Pr(\mathbf{\tilde{U}} = \mathbf{U}_j)} 
			\nonumber \\
			&= \Pr(\mathbf{\tilde{U}} = \mathbf{U}), \nonumber
			\end{align}
			where (a) follows form Bayesian Rule, (b) follows from the fact that according to Corollary \ref{natije asli randomness}, each row of matrix $\mathbf{\tilde{R}}$ has a uniform distribution over $\mathbb{F}^ {p \times (t-1)q}$, thus $\mathbf{\tilde{R}}$ has a uniform distribution over $\mathbb{F}^ {p \times (t-1)q}$, which implies that $\Pr(\mathbf{\tilde{R}} = \mathbf{T} -  \mathbf{U}) = \Pr(\mathbf{\tilde{R}} = \mathbf{T} -  \mathbf{U}_j) = \frac{1}{|\mathbb{F}|^{mp(t-1)q}}$. 
			Thus, we have $H(\mathbf{\tilde{U}}|\mathbf{\tilde{T}})= H(\mathbf{\tilde{U}})$. Therefore $I(\mathbf{\tilde{U}};\mathbf{\tilde{T}}) = 0$.
			
			On the other hand, from the definition of $\mathcal{A}$, one can see that $\mathcal{A} \rightarrow \mathbf{\tilde{U}} \rightarrow \mathbf{\tilde{T}}$ is a Markov chain. Thus, according to data processing inequality, we have
			\begin{align*}
			I(\mathcal{A};\mathbf{\tilde{T}}) \leq I(\mathbf{\tilde{U}};\mathbf{\tilde{T}}) = 0.
			\end{align*}
			Therefore, if the adversaries know the elements of $\mathbf{\tilde{T}}$, they are not able to gain any information about the elements of $\mathcal{A}$.
	\end{proof}
	
	\begin{corollary} \label{khafan}
	 	Assume that 
	 	\begin{align*}
	 	\mathbf{Y} \defeq \begin{bmatrix}
	 	\mathbf{Y}_1, \mathbf{Y}_2, \dots, \mathbf{Y}_k
	 	\end{bmatrix},
	 	\end{align*}
	 	where $\mathbf{Y}_1, \mathbf{Y}_2, \dots, \mathbf{Y}_k \in \mathbb{F}^{m \times \frac{m}{k}}$. Assume that $\mathbf{Y}$ is shared using polynomial function
	 	\begin{align*}
	 	\mathbf{F}_{\mathbf{Y},b,t,k}(x) = \sum_{j=1}^{k}\mathbf{Y}_nx^{b(j-1)} + \sum_{j=1}^{t-1} \mathbf{R}_jx^{k^2+j-1},
	 	\end{align*}
	 	where $\mathbf{R}_1, \mathbf{R}_1, \dots, \mathbf{R}_{t-1}$ are chosen independently and uniformly at random in  $\mathbb{F}^{m \times \frac{m}{k}}$, with $N$ workers, i.e., $\mathbf{F}_{\mathbf{Y},b,t,k}(\alpha_n)$ is delivered to worker $n$, $n \in [N]$. For any subset $\mc{S} \subset [N], |\mc{S}| = t-1$, we have
	 	\begin{align*}
	 	H(\mathbf{Y}| \mathbf{F}_{\mathbf{Y},b,t,k}(\alpha_i), i \in \mc{S}) &= H(\mathbf{Y}), \\
	 	H(\mathbf{F}_{\mathbf{Y},b,t,k}(\alpha_i), i \in \mc{S}|\mathbf{Y})  &= H(\mathbf{F}_{\mathbf{Y},b,t,k}(\alpha_i), i \in \mc{S}).
	 	\end{align*}
	 \end{corollary}
	\begin{proof}
		The proof follows directly from Lemma \ref{lemme security}.
	\end{proof}
	
	 Intuitively, we can explain the privacy constraints as follows. If we consider any subset of $t-1$ workers, each share that they receive from the sources or any other workers includes contribution $t-1$ random matrices, excluding the original data itself. Thus, if we ignore the original data, the number of equations and the number of variables are the same. However, because of original data, the number of equations is always less than the number of the variables, no matter how data is involved in these equations. Thus, the adversaries cannot gain any information about the private inputs.
	 
	 Now we formally prove the privacy constraints \eqref{privacy for workers} and \eqref{privacy for master}, for Algorithm \eqref{LSMPC}. Assume that the semi-honest workers are $i_1, i_2, \dots, i_{t-1} \in [N]$. In order to prove constraint \eqref{privacy for workers}, for Algorithm \ref{LSMPC}, we must show that for any $\mathbf{X}^{[1]}, \mathbf{X}^{[2]} \dots, \mathbf{X}^{[\Gamma]} \in \mathbb{F}^{m \times m}$, and polynomial function $\mathbf{G}: (\mathbb{F}^{m \times m})^\Gamma \rightarrow \mathbb{F}^{m \times m}$ we have 
	 \begin{align*}
		H(\mathbf{X}^{[j]}, j\in [\Gamma]| \bigcup_{n \in \mc{S}} \{\mc{M}_{n^{\prime} \to n}, n^{\prime} \in [N]\}, \mathbf{\tilde{X}}_{\gamma n}, \gamma \in [\Gamma], n \in \mc{S})  = H(\mathbf{X}^{[j]}, j\in [\Gamma]), 
	 \end{align*}
	 where $\mc{S} = \{i_1, i_2, \dots, i_{t-1}\}$. Assume that the calculations is done through  $R$ rounds. Let us define $\mathcal{M}^{(r)}_{n^{\prime} \to n}$ as the messages sent from  worker $n^{\prime}$ to worker $n$, in round $r \in [R]$. Thus
	 \begin{align*}
	 	\mc{M}_{n^{\prime} \to n} = \bigcup_{r=1}^R \mc{M}^{r}_{n^{\prime}\to n}.
	 \end{align*}
	
	Also define $\mc{R}^{(r)}_{n^{\prime} \to n}$ as the set of all random matrices that worker $n^{\prime}$ uses for sending a message to worker $n$, in round $r$. Let us define
	\begin{align*}
		\mc{M}^{(r)}_\mathcal{S} &\defeq \bigcup_{n^{\prime} \in [\Gamma], n \in \mathcal{S}} 
		\mc{M}^{(r)}_{n^{\prime}\to n},  \\
		\mc{R}^{(r)}_\mathcal{S} &\defeq \bigcup_{n^{\prime} \in [\Gamma], n \in \mathcal{S}}  \mc{R}^{(r)}_{n^{\prime} \to n},  \\ 
		\mc{X} &\defeq \{\mathbf{X}^{[1]}, \mathbf{X}^{[2]} \dots, \mathbf{X}^{[\Gamma]}\},  \\
		\mc{\tilde{X}}_\mathcal{S} &\defeq \{\mathbf{\tilde{X}}_{\gamma n}, \gamma \in [\Gamma], n \in \mc{S}\}.
	\end{align*}
	From the definition, we have 
	\begin{align*}
		& H(\mathbf{X}^{[j]}, j\in [\Gamma]| \bigcup_{n \in \mc{S}} \{\mc{M}_{n^{\prime} \to n}, n^{\prime} \in [N]\}, \mathbf{\tilde{X}}_{\gamma n}, \gamma \in [\Gamma], n \in \mc{S})  = \\
		& H(\mc{X}| \mc{\tilde{X}}_\mathcal{S}, \mc{M}^{(r)}_\mathcal{S}, r \in [R]).
	\end{align*}
	Therefore, to prove the privacy constraint \eqref{privacy for workers}, it is sufficient to show that
	\begin{align*}
		I(\mc{\tilde{X}}_\mathcal{S}, \mc{M}^{(1)}_\mathcal{S}, \mc{M}^{(2)}_\mathcal{S}, \dots, \mc{M}^{(R)}_\mathcal{S}; \mc{X}) = 0.
	\end{align*} 
	
	One can see that 
	\begin{align} \label{s1}
		H(\mc{\tilde{X}}_\mathcal{S}, \mc{M}^{(1)}_\mathcal{S}, \mc{M}^{(2)}_\mathcal{S}, \dots, \mc{M}^{(R)}_\mathcal{S}| \mc{X}) &= H(\mc{\tilde{X}}_\mathcal{S}|\mc{X}) + H(\mc{M}^{(1)}_\mathcal{S}| \mc{\tilde{X}}_\mathcal{S}, \mc{X})  \\ 
		&+ H(\mc{M}^{(2)}_\mathcal{S}|\mc{M}^{(1)}_\mathcal{S}, \mc{\tilde{X}}_\mathcal{S}, \mc{X}) \nonumber \\ 
		&\vdots \nonumber \\
		&+ H(\mc{M}^{(R)}_\mathcal{S}|\mc{M}^{(R-1)}_\mathcal{S}, \dots, \mc{M}^{(1)}_\mathcal{S}, \mc{\tilde{X}}_\mathcal{S}, \mc{X}). \nonumber
	\end{align}
	
	According to  Corollary \ref{khafan}, we have $H(\mc{\tilde{X}}_\mathcal{S}|\mc{X}) =  H(\mc{\tilde{X}}_\mathcal{S})$. In addition, since in each round $r \in [R]$, $\mc{M}^{(r)}_\mathcal{S}$ is a function of $\mc{\tilde{X}}_\mathcal{S}$, $\mc{M}^{(1)}_\mathcal{S}$, $\mc{M}^{(2)}_\mathcal{S}$, $\mc{M}^{(r-1)}_\mathcal{S}$, and $\mc{R}^{(r)}_\mathcal{S}$, then
	\begin{align}\label{s2}
		H(\mc{M}^{(r)}_\mathcal{S}|\mc{M}^{(r-1)}_\mathcal{S}, \dots, \mc{M}^{(1)}_\mathcal{S}, \mc{\tilde{X}}_\mathcal{S}, \mc{X}) & = H(\mc{R}^{(r)}_{\mc{S}}|\mc{M}^{(r-1)}_\mathcal{S}, \dots, \mc{M}^{(1)}_\mathcal{S}, \mc{\tilde{X}}_\mathcal{S}, \mc{X}) \\
		&= H(\mc{R}^{(r)}_{\mc{S}}) \overset{(a)}\geq H(\mc{M}^{(r)}_{\mc{S}}), \nonumber
	\end{align}
	where (a) follows from the fact that $H(\mc{R}^{(r)}_{\mc{S}})$ and $H(\mc{M}^{(r)}_{\mc{S}})$ have the same size and $\mc{R}^{(r)}_{\mc{S}}$ has a uniform distribution. According to \eqref{s1} and \eqref{s2} we have
	\begin{align*}
		H(\mc{\tilde{X}}_\mathcal{S}, \mc{M}^{(1)}_\mathcal{S}, \mc{M}^{(2)}_\mathcal{S}, \dots, \mc{M}^{(R)}_\mathcal{S}| \mc{X}) &\geq H(\mc{\tilde{X}}_\mathcal{S}) + H(\mc{M}^{(1)}_\mathcal{S}) + H(\mc{M}^{(2)}_\mathcal{S})+ \dots + H(\mc{M}^{(R)}_\mathcal{S})  \\
		&\geq H(\mc{\tilde{X}}_\mathcal{S}, \mc{M}^{(1)}_\mathcal{S}, \mc{M}^{(2)}_\mathcal{S}, \dots, \mc{M}^{(R)}_\mathcal{S}).
	\end{align*}
	Therefore, 
	\begin{align*}
		I(\mc{\tilde{X}}_\mathcal{S}, \mc{M}^{(1)}_\mathcal{S}, \mc{M}^{(2)}_\mathcal{S}, \dots, \mc{M}^{(R)}_\mathcal{S}; \mc{X}) \leq 0.
	\end{align*}
	Thus
	\begin{align*}
		I(\mc{\tilde{X}}_\mathcal{S}, \mc{M}^{(1)}_\mathcal{S}, \mc{M}^{(2)}_\mathcal{S}, \dots, \mc{M}^{(R)}_\mathcal{S}; \mc{X}) = 0,
	\end{align*}
	 which proves the privacy constraint at the workers \eqref{privacy for workers}.
	 
	Now we prove  constraint (\ref{privacy for master}). Assume that 
	\begin{align*}
		\mathbf{Y} = \begin{bmatrix}
			{\mathbf{Y}_{1}, \mathbf{Y}_{2}, \dots, \mathbf{Y}_{k}}
		\end{bmatrix},
	\end{align*}
	where $\mathbf{Y}_{1}, \mathbf{Y}_{2}, \dots, \mathbf{Y}_{k} \in \mathbb{F}^{m \times \frac{m}{k}}$.
	Since the result at the master is in the form of \emph{polynomial sharing}, according to \eqref{tarife raveshe sharing} there exist a polynomial function $\mathbf{F}_{\mathbf{Y},b,t,k}(x)$, where $\mathbf{F}_{\mathbf{Y},b,t,k}(\alpha_{i}) = \mathbf{O}_i, i \in [N]$. More precisely, according to \eqref{tarife raveshe sharing}, there are $\mathbf{R}_1, \mathbf{R}_2, \dots, \mathbf{R}_{t-1}$ with independent and uniform distribution in  $\mathbb{F}^{m \times \frac{m}{k}}$, where 
	\begin{align}
		\mathbf{F}_{\mathbf{Y},b,t,k}(x) &= \sum_{j=1}^{k}\mathbf{Y}_{j}x^{b(j-1)} + \sum_{j=1}^{t-1} \mathbf{R}_jx^{k^2+j-1}, \nonumber \\
		\mathbf{F}_{\mathbf{Y},b,t,k}(\alpha_{n}) &= \mathbf{O}_n. \label{setare}
	\end{align}
	One can see that according to Theorem \ref{qazie sharing}, we have 
	\begin{align}
		H(\mathbf{Y}|\mathbf{O}_1, \mathbf{O}_2, \dots, \mathbf{O}_N) &= 0, \label{a1} \\
		H(\mathbf{Y}_{1}, \mathbf{Y}_{2}, \dots, \mathbf{Y}_{k}, \mathbf{R}_1, \mathbf{R}_2, \dots, \mathbf{R}_{t-1}| \mathbf{O}_1, \mathbf{O}_2, \dots, \mathbf{O}_N)&=0, \label{a3}\\
		H(\mathbf{O}_1, \mathbf{O}_2, \dots, \mathbf{O}_N|\mathbf{Y}_{1}, \mathbf{Y}_{2}, \dots, \mathbf{Y}_{k}, \mathbf{R}_1, \mathbf{R}_2, \dots, \mathbf{R}_{t-1}) &= 0 \label{a2}. 
	\end{align}
	 Thus, we have
	\begin{align*}
		H(\mc{X}| \mathbf{Y}, \mathbf{O}_1, \mathbf{O}_2, \dots, \mathbf{O}_N) 
		 &\overset{(a)} = H(\mc{X}| \mathbf{O}_1, \mathbf{O}_2, \dots, \mathbf{O}_N) \\ 
		 & \overset{(b)}= H(\mc{X}| \mathbf{O}_1, \mathbf{O}_2, \dots, \mathbf{O}_N, \mathbf{Y}_{1}, \mathbf{Y}_{2}, \dots, \mathbf{Y}_{k}, \mathbf{R}_1, \mathbf{R}_2, \dots, \mathbf{R}_{t-1}) \\
		 & \overset{(c)} = H(\mc{X}| \mathbf{Y}_{1}, \mathbf{Y}_{2}, \dots, \mathbf{Y}_{k}, \mathbf{R}_1, \mathbf{R}_2, \dots, \mathbf{R}_{t-1}) \\
		 & \overset{(d)} = H(\mc{X}| \mathbf{Y}_{1}, \mathbf{Y}_{2}, \dots, \mathbf{Y}_{k}) \\
		 &‌= H(\mc{X}|\mathbf{Y})
	\end{align*}
	where (a) follows from \eqref{a1}, (b) and (c) follow from \eqref{setare}, \eqref{a3}, and \eqref{a2}, and (d) follows from the fact that $\mathbf{R}_1, \mathbf{R}_2, \dots, \mathbf{R}_{t-1}$ have independent and uniform distribution in  $\mathbb{F}^{m \times \frac{m}{k}}$ and independent from $\mc{X}$.

	\clearpage
	
	\section{Arithmetic Circuits} \label{Arithmetic Circuits}
	In this part, we describe some rules to create the arithmetic circuit corresponding to a function $\\ \mathbf{G}(\mathbf{X}^{[1]}, \mathbf{X}^{[2]} \dots, \mathbf{X}^{[\Gamma]})$ and a specific order of computation. We know that each polynomial function can be written as a sum of production terms
		\begin{align}\label{tabe G}
		\mathbf{G}(\mathbf{X}^{[1]}, \mathbf{X}^{[2]} \dots, \mathbf{X}^{[\Gamma]}) = \sum_{j=1}^{M} \mathbf{G}_j(\mathbf{X}^{[1]}, \mathbf{X}^{[2]} \dots, \mathbf{X}^{[\Gamma]}),
		\end{align}
		where $M$ is the number of the monomial terms, and $\mathbf{G}_j$'s are monomial functions, for $j \in [M]$.
		\begin{example}
			\label{notation}
			Assume that the desired function $\mathbf{G}$ is
			\begin{align}
			\mathbf{G}(\mathbf{X}^{[1]}, \mathbf{X}^{[2]}, \mathbf{X}^{[3]}, \mathbf{X}^{[4]}) = (\mathbf{X}^{[2]})^{T}(\mathbf{X}^{[1]})^{2}\mathbf{X}^{[3]} + \mathbf{X}^{[2]}\mathbf{X}^{[4]}(\mathbf{X}^{[3]})^{T}. \nonumber
			\end{align}
			
			According to the notations we have
			\begin{align}
			\mathbf{G}_1(\mathbf{X}^{[1]}, \mathbf{X}^{[2]}, \mathbf{X}^{[3]}, \mathbf{X}^{[4]}) = (\mathbf{X}^{[2]})^{T}(\mathbf{X}^{[1]})^{2}\mathbf{X}^{[3]}, \nonumber \\
			\mathbf{G}_2(\mathbf{X}^{[1]}, \mathbf{X}^{[2]}, \mathbf{X}^{[3]}, \mathbf{X}^{[4]}) = \mathbf{X}^{[2]}\mathbf{X}^{[4]}(\mathbf{X}^{[3]})^{T}. \nonumber
			\end{align}
		\end{example}
		There are some rules to represent $\Gb$ by multiplication and addition gates. The rules are described in the following. For simplification we show a multiplication gate by an AND gate, and an addition gate by an OR gate.
		\begin{enumerate}
			\item \emph{Rule 1:} \label{zarbi}
			
			Assume that the function is in the form of 	$\prod_{j=1}^{\Gamma} \mathbf{X}^{[j]}$. In order to represent this function, first we multiply the last two matrices ($\mathbf{X}^{[\Gamma -1]}, \mathbf{X}^{[\Gamma]}$), then we multiply $\mathbf{X}^{[\Gamma-2]}$ to the result of the previous operation and so on.  The order of computation is shown in Fig. \ref{multiplicationGate}. 
			
			\begin{figure}[htbp]
				\centering
				\includegraphics[width=10cm]{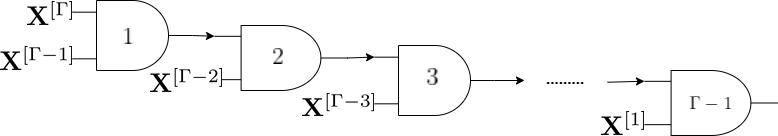}
				\caption{The circuit representing the order of operations in calculating the function $\mathbf{G} = \prod_{j=1}^{\Gamma} \mathbf{X}^{[j]}$.}
				\label{multiplicationGate}
			\end{figure}
			
			\item \emph{Rule 2:} \label{jam}
			
			Assume that the function is in the form of $\sum_{j=1}^{\Gamma} \mathbf{X}^{[j]}$. In order to represent this function, first we add the  last two matrices ($\mathbf{X}^{[\Gamma -1]} + \mathbf{X}^{[\Gamma]}$), then we add  $\mathbf{X}^{[\Gamma-2]}$ to the result of the previous operation, and so on. The order of computation is shown in Fig. \ref{additionGate}. 
			\begin{figure}[!htbp]
				\centering
				\includegraphics[width=10cm]{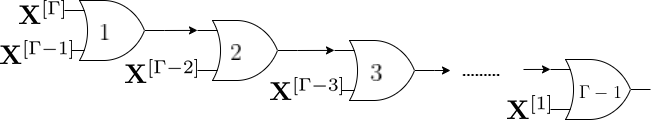}
				\caption{The circuit representing the order of operations in calculating the function  $\mathbf{G} = \sum_{j=1}^{\Gamma} \mathbf{X}^{[j]}$}
				\label{additionGate}
			\end{figure}
			\item \emph{Rule 3:} \label{koli}
			To represent a general function \eqref{tabe G}, and assign a specific order to the computations, we first compute $\mathbf{G}_M$ based on Rule \ref{zarbi}, keep the result, and compute $\mathbf{G}_{M-1}$ based on Rule \ref{zarbi}, and add up the result based on rule \ref{jam}, and then compute $\mathbf{G}_{M-2}$ based on Rule \ref{zarbi}, and so on. The representation and order of computation are shown for an example in Fig.\ref{circuitShape}.
			\begin{figure}[!htbp]
				\centering
				\includegraphics[width=10cm]{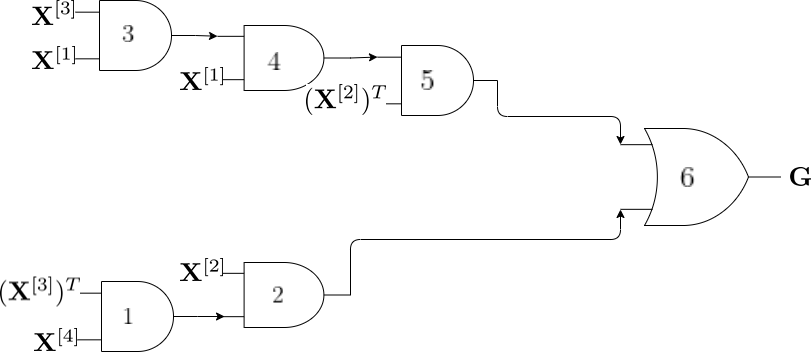}
				\caption{The circuit representing the order of computation for function $\mathbf{G}(\mathbf{X}^{[1]}, \mathbf{X}^{[2]}, \mathbf{X}^{[3]}, \mathbf{X}^{[4]}) = (\mathbf{X}^{[2]})^{T}(\mathbf{X}^{[1]})^{2}\mathbf{X}^{[3]} + \mathbf{X}^{[2]}\mathbf{X}^{[4]}(\mathbf{X}^{[3]})^{T}$.}
				\label{circuitShape}
			\end{figure}
	\end{enumerate}

\end{appendices}
\end{document}